\documentclass{article}

\usepackage{arxiv}
\usepackage[numbers]{natbib}

\usepackage[utf8]{inputenc} 
\usepackage[T1]{fontenc}    
\usepackage{hyperref}       
\usepackage{url}            
\usepackage{booktabs}       
\usepackage{amsfonts, amsmath, amsthm, amssymb}       
\usepackage{nicefrac}       
\usepackage{microtype}      
\usepackage{lipsum}
\usepackage{fancyhdr}       
\usepackage{graphicx}       

\usepackage{hyperref}
\usepackage{graphicx}
\usepackage{pdfpages}

\usepackage{tikz}
\usepackage{booktabs}
\usepackage{caption}
\usepackage{multirow}
\usepackage{subcaption} 
\usepackage{tablefootnote}
\usepackage{tikz}
\usepackage{algorithm}
\usepackage[noend]{algpseudocode}

\algnewcommand\algorithmicforeach{\textbf{for each}}
\algdef{S}[FOR]{ForEach}[1]{\algorithmicforeach\ #1\ \algorithmicdo}
\newtheorem{theorem}{Theorem}[section]
\newtheorem{lemma}[theorem]{Lemma}
\newtheorem{definition}[theorem]{Definition}
\newtheorem{corollary}[theorem]{Corollary}
\newtheorem*{remark}{Remark}

\newcommand{\z}{\boldsymbol{z}}
\newcommand{\M}{\mathcal{M}}
\newcommand{\X}{\mathcal{X}}

\newcommand{\LL}{\mathcal{L}}
\newcommand{\MinPts}{\text{MinPts}}
\newcommand{\Span}{\operatorname{Span}}
\newcommand{\op}[1]{\operatorname{#1}}
\newcommand{\ARI}{\operatorname{ARI}}
\newcommand{\AMI}{\operatorname{AMI}}
\newcommand{\dist}{\mathrm{dist}}

\title{Approximate DBSCAN under Differential Privacy}

\author{Yuan Qiu\\
CNRS@CREATE\\
\texttt{yuan.qiu@cnrsatcreate.sg}\\
\And
Ke Yi\\
Hong Kong University of Science and Technology\\
\texttt{yike@cse.ust.hk}\\
}

\begin{document}

\maketitle


\begin{abstract}
This paper revisits the DBSCAN problem under differential privacy (DP).  Existing DP-DBSCAN algorithms aim at publishing the cluster labels of the input points.  However, we show that both empirically and theoretically, this approach cannot offer any utility in the published results.  We therefore propose an alternative definition of DP-DBSCAN based on the notion of spans.  We argue that publishing the spans actually better serves the purposes of visualization and classification of DBSCAN.  Then we present a linear-time DP-DBSCAN algorithm achieving the sandwich quality guarantee in any constant dimensions, as well as matching lower bounds on the approximation ratio. A key building block in our algorithm is a linear-time algorithm for constructing a histogram under pure-DP, which is of independent interest.  Finally, we conducted experiments on both synthetic and real-world datasets to verify the practical performance of our DP-DBSCAN algorithm.
\end{abstract}

\section{Introduction}

\subsection{DBSCAN}
As one of the most popular clustering algorithms, {DBSCAN} (Density-Based Spatial Clustering of Applications with Noise)~\cite{DBLP:conf/kdd/EsterKSX96} has been receiving continuous attention~\cite{DBLP:conf/pakdd/CampelloMS13,DBLP:journals/tkdd/CampelloMZS15,DBLP:conf/sigmod/GanT15,DBLP:journals/tods/SchubertSEKX17, DBLP:journals/pacmmod/MoSD24} over the past 30 years. Unlike centroid-based clustering algorithms such as k-means~\cite{DBLP:journals/tit/Lloyd82}, which tend to identify clusters in spherical shapes,  DBSCAN can detect clusters of arbitrary shapes. Additionally, DBSCAN offers several advantages, such as not requiring the number of clusters to be specified in advance and its robustness to outliers.

Without loss of generality, we consider the $d$-dimensional unit cube $\LL:=[0,1]^d$ as the feature space, equipped with Euclidean distance $\dist(\cdot, \cdot)$.  
Two parameters decide the DBSCAN clustering result: a real number\footnote{
We use $\alpha$ instead of the usual symbol $\varepsilon$, which will be used to denote the privacy parameter. } $\alpha\in (0,1)$ and an integer $\MinPts \ge 1$.  For any location $l\in \LL$, its $\alpha$-neighborhood is defined as $\mathbf{B}(l,\alpha):=\{l'\in \LL: \dist(l,l')<\alpha\}$, the $d$-dimensional ball centered at $l$ with radius $\alpha$.  Let $P=(p_1,\dots,p_n) \in \LL^n$ be the input point set.  
A point $p\in P$ is called a \textit{core point}, if its $\alpha$-neighborhood contains at least $\MinPts$ points of $P$, 
otherwise a \textit{noise point}.
Two core points $p$ and $q$ are $\alpha$-reachable if $\dist(p,q)<\alpha$; they are $\alpha$-connected if they are directly or transitively $\alpha$-reachable.   Then an $(\alpha,\MinPts)$-cluster $C$ is a nonempty maximal subset of the core points that are mutually $\alpha$-connected, and the DBSCAN output, denoted $\mathcal{C}(\alpha,\MinPts)$, consists of all the $(\alpha,\MinPts)$-clusters.  Note that the noise points do not belong to any cluster.

\subsection{Approximate DBSCAN}
In order to reduce the computational cost, \citet{DBLP:conf/sigmod/GanT15} have introduced  \textit{approximate DBSCAN} based on the \textit{sandwich quality guarantee}.  For our purpose which will become clear shortly, we generalize their definition by introducing a second approximation parameter $\tau$ (their definition is the special case $\tau=0$):

\begin{definition}[Sandwich Quality Guarantee]\label{def:sand}
For $\rho>1$ and $0\leq \tau < \MinPts$, a set of clusters $\hat{\mathcal{C}}$ are called $(\rho,\tau)$-approximate clusters, if the following conditions are satisfied:
\begin{enumerate}
\item For each cluster $C_1\in\mathcal{C}(\alpha,\MinPts)$, there exists an approximate cluster $\hat{C}\in \hat{\mathcal{C}}$ such that $C_1\subseteq \hat{C}$.
\item For each approximate cluster $\hat{C}\in \hat{\mathcal{C}}$, there exists a cluster $C_2\in\mathcal{C}(\rho\alpha, \MinPts-\tau)$ such that $\hat{C}\subseteq C_2$.
\end{enumerate}
\end{definition}

The approximation parameter $\rho$ is critical in the computational complexity of DBSCAN: For any constant $\rho>1$ and constant $d$, they present an $O(n)$-time algorithm that returns $(\rho,0)$-approximate clusters for any $P$, while giving an $\Omega(n^{4/3})$ lower bound for $\rho=1, \tau=0$ (i.e., accurate clusters) for any $d\ge 3$ \cite{DBLP:conf/sigmod/GanT15}. They also empirically demonstrate that approximate clusters are very close to the true clusters on most realistic datasets. 

\subsection{DBSCAN under Differential Privacy}

As with other learning algorithms, clustering results can contain sensitive information that may expose personal data, posing significant privacy risks.  Unsurprisingly, DBSCAN has been studied under \textit{differential privacy (DP)}~\cite{DBLP:journals/fttcs/DworkR14}, the \textit{de facto} standard for personal data privacy. Since publishing any of the input points in $P$ cannot possibly satisfy DP, so existing work~\cite{DBLP:conf/iiki/NiLLBY17,wu2015dp,DBLP:journals/access/NiLWJY18,jin2019improved} choose to output only the cluster labels, i.e., the DP-DBSCAN output is a vector $(c_1,\dots, c_n)$ such that $p_i$ and $p_j$ are in the same cluster iff $c_i=c_j>0$, while $c_i=0$ indicates that $p_i$ is a noise point.  However, we tested the existing DP-DBSCAN algorithms and find that they offer little utility. In all our test datasets, all these algorithms classified all points into a single cluster, which is clearly not informative. Our first result in this paper is a negative result, showing that this lack of utility is not due to the algorithm design, but the output format.  Specifically, we prove that, if required to output the cluster labels of the input points, then no DP algorithm can achieve any finite $\rho$ and any $\tau < \MinPts$ with constant probability, i.e., no non-trivial sandwich quality guarantee is achievable. 

\begin{figure*}[htbp]
\centering
     \begin{subfigure}[t]{0.24\textwidth}
         \includegraphics[width=\textwidth]{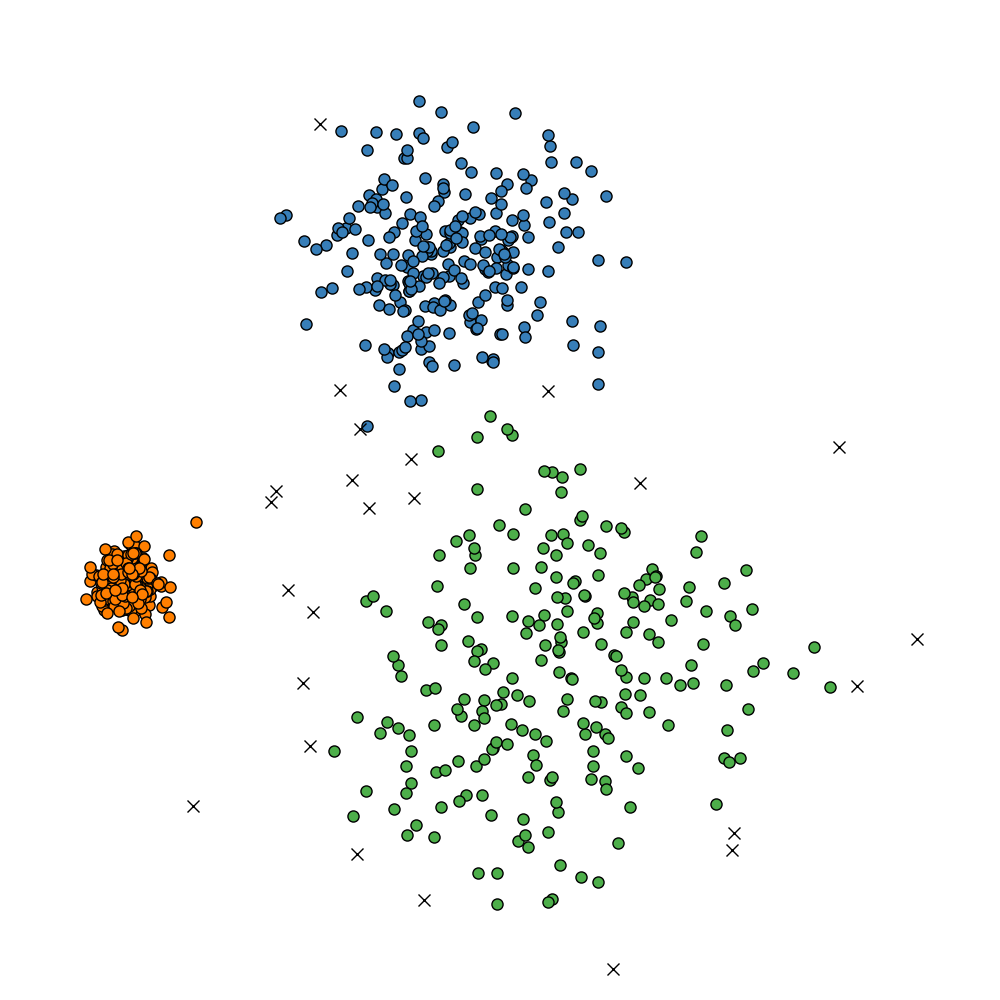}
         \caption{Clusters (Blobs)}
         \label{fig:dbscanBlob}
     \end{subfigure}    \hfill
     \begin{subfigure}[t]{0.24\textwidth}
         \centering
         \includegraphics[width=\textwidth]{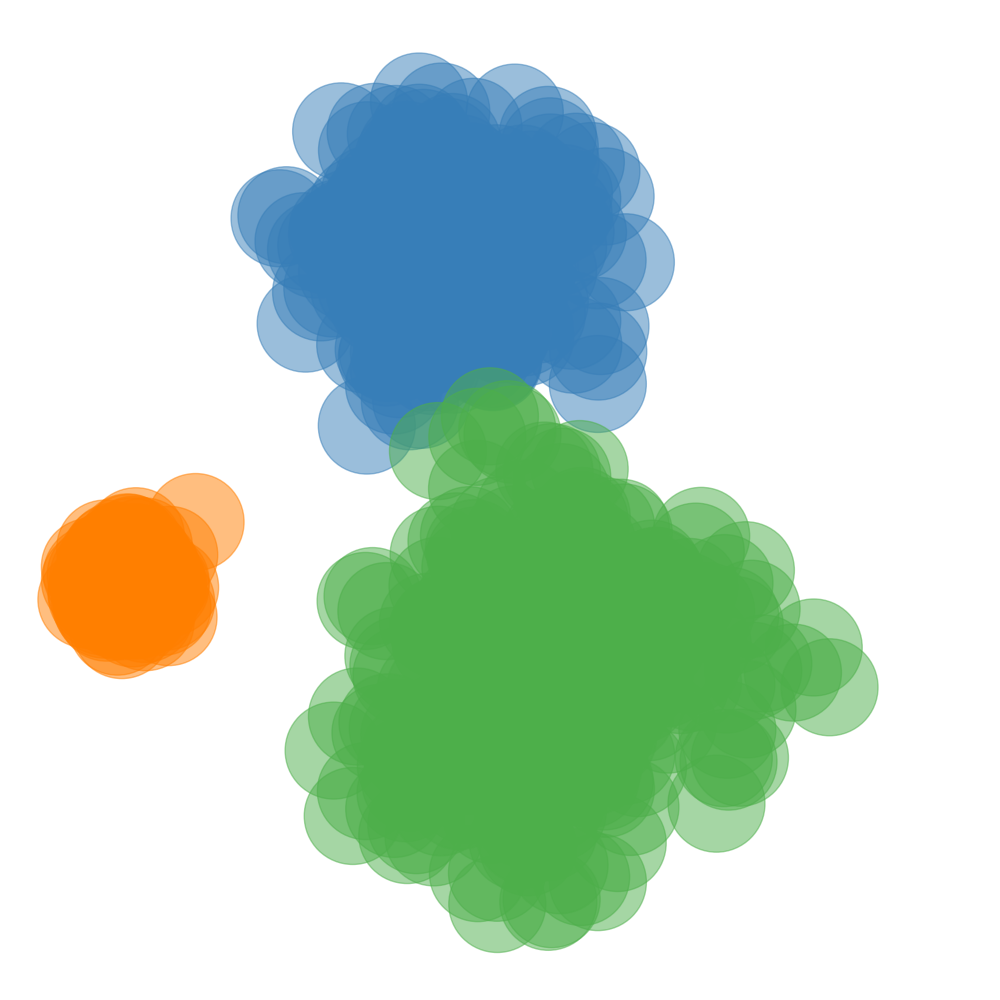}
         \caption{Cluster spans (Blobs)}
         \label{fig:spanBlob}
     \end{subfigure}     \hfill
     \begin{subfigure}[t]{0.24\textwidth}
         \centering
         \includegraphics[width=\textwidth]{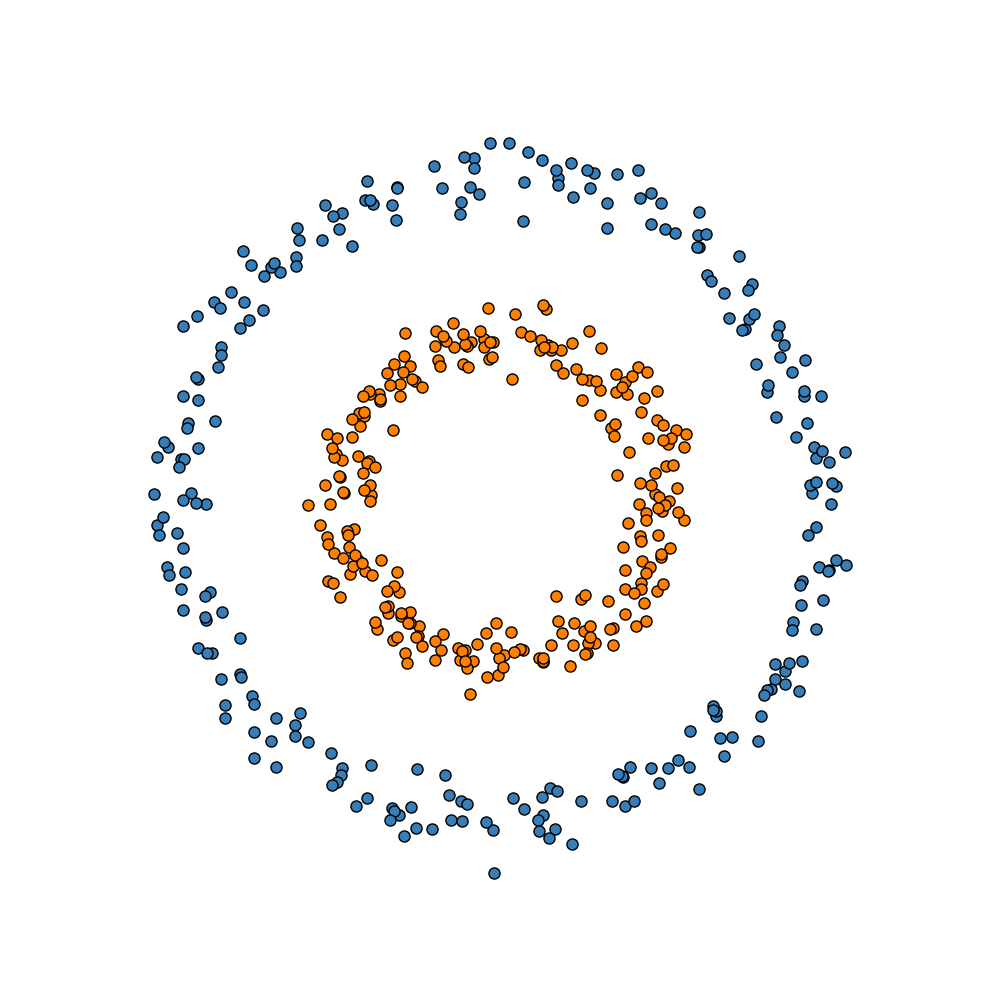}
         \caption{Clusters (Circles)}
         \label{fig:dbscanCirc}
     \end{subfigure}     \hfill
     \begin{subfigure}[t]{0.24\textwidth}
         \centering
         \includegraphics[width=\textwidth]{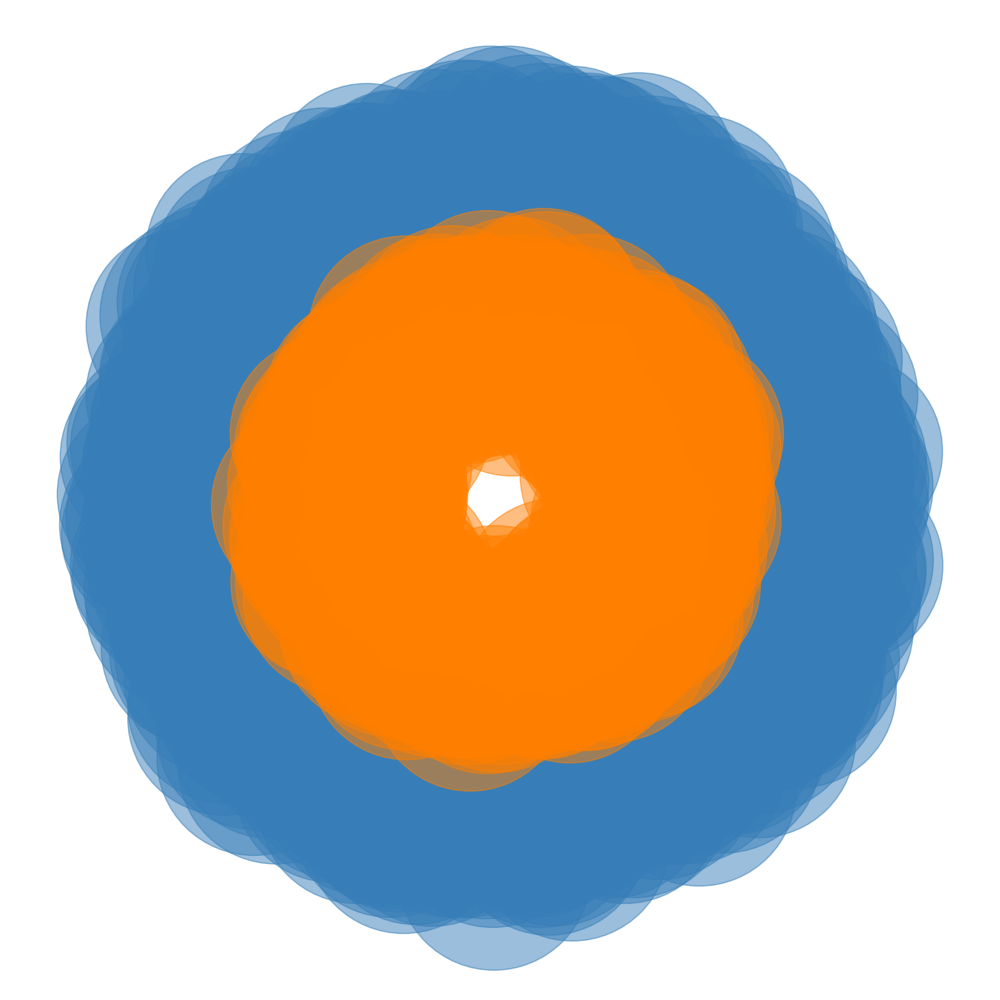}
         \caption{Cluster spans (Circles)}
         \label{fig:spanCirc}
     \end{subfigure}
\caption{Span of clusters.}
\label{fig:span}
\end{figure*}

To get around this negative result, we propose an alternative output format for DP-DBSCAN: For 
a cluster $C$, we define its \textit{span} as $\Span(C):=\cup_{p\in C} \mathbf{B}(p,\alpha)$, and we aim at outputting the spans of the clusters; please see some illustrations in Figure \ref{fig:span}.  Compared with the cluster labels, the spans do not reveal how the input points themselves are clustered, which is impossible by our negative result.  However, they serve the more useful functions of DBSCAN: visualization and classification.  The visualization result is illustrated in Figure \ref{fig:span}; more plots are provided in Section \ref{sec:exp}.  To classify a new  point in the feature space, one can simply check which span the given point falls into, or declare it as a noise point if it does not fall into any span.  Note that the spans may overlap a little bit, but the overlapping region must have width less than $\alpha$, and points falling into the overlapping region can be classified into either cluster. 
Essentially, these spans serve as a classification model, similar to the $k$ centroids in k-means clustering.  

More importantly, we are able to design DP algorithms that output approximate spans that satisfy the sandwich quality guarantee:
 
\begin{definition}[Sandwich Quality Guarantee for Spans]\label{def:ratio}
For $\rho>0$ and $0\leq \tau<\MinPts$, a collection of spans $\mathcal{S}=\{S_1,\dots,S_m\}$ are called $(\rho,\tau)$-approximate spans, if the following conditions are satisfied:
\begin{enumerate}
\item For each cluster $C_1\in\mathcal{C}(\alpha,\MinPts)$, there exists a span $S\in \mathcal{S}$ such that $C_1\subseteq S$.
\item For each span $S\in \mathcal{S}$, there exists a cluster $C_2\in\mathcal{C}(\rho\alpha,\MinPts-\tau)$, such that $S\subseteq \Span(C_2)$.
\end{enumerate}
\end{definition}

\subsection{Our Contributions}
Specifically, we make the following contributions in this paper with respect to DP-DBSCAN:

\paragraph{(1) A new definition of approximate DP-DBSCAN (Section \ref{sec:span})}  We prove that, using the existing definition of DP-DBSCAN, it is impossible to achieve any meaningful utility guarantee.  We then propose a span-based new definition, which not only allows us to achieve the sandwich quality guarantee, but also serves visualization and classification purposes of clustering.  

\paragraph{(2) Approximation lower bounds (Section \ref{sec:lower})}
For the best achievable span-based sandwich quality guarantee, we first prove two lower bounds: $\rho\ge 3$ and $\tau=\Omega(\frac{1}{\varepsilon} \log{\frac{1}{\rho \alpha}})$.  Both lower bounds hold for any DP-DBSCAN algorithm even with unlimited computing power, which implies that the privacy constraint incurs a higher level of hardness for DBSCAN than the running time requirement (recall that in the non-private setting, one can achieve $\rho=1,\tau=0$ with an $O(n^2)$-time algorithm, and $\rho>1,\tau=0$ with an $O(n)$-time algorithm). 

\paragraph{(3) Approximate DP-DBSCAN (Section \ref{sec:upper})}
We design an $O(n)$-time DP-DBSCAN mechanism that achieves any constant $\rho>3$ and $\tau=O(\frac{1}{\varepsilon}\log\frac{1}{\alpha})$.  This approaches the lower bound on  $\rho$ within an arbitrarily small constant, while matching the lower bound on $\tau$ up to a constant factor.  

\paragraph{(4) Linear-time pure-DP histogram (Section \ref{sec:hist})}
A key building block of our DP-DBSCAN algorithm is a differentially private histogram, which is itself a fundamental problem in DP.  A naive method simply adds a Laplace noise to the frequency of each element in the universe $\X$, regardless they appear in the given dataset or not.  This results in $O(|\X|)$ time, which does not work for a large $\X$.  \citet{DBLP:conf/icdt/CormodePST12} designed an $O(n)$-time histogram algorithm under pure-DP, but their proof did not show that the bins in the histogram are independent.  We give a complete proof on both the privacy and utility of their pure-DP histogram. 

\paragraph{(5) Practical optimizations and experimental study (Section \ref{sec:discuss}-\ref{sec:exp})}
In addition to providing the theoretical guarantees above, we have also made some practical optimizations to our DP-DBSCAN algorithm.  Finally, we conducted extensive experiments on synthetic and real-world datasets to verify its practical performance.

\subsection{Other Related Work}

Besides DBSCAN, some other clustering problems have also been studied under DP, particularly k-means~\cite{DBLP:conf/codaspy/SuCLBJ16,DBLP:conf/icml/BalcanDLMZ17,DBLP:conf/pods/HuangL18,DBLP:conf/nips/StemmerK18,DBLP:conf/nips/Ghazi0M20,DBLP:conf/icml/ChangG0M21,DBLP:conf/aaai/NguyenCX21,DBLP:conf/aaai/JonesNN21,DBLP:conf/icml/CohenKMST21}. In DP-Kmeans, a mechanism publishes $k$ privatized centroids that can be used for classifying new points through a nearest-neighbor query over the centroids. The objective in designing DP-Kmeans mechanisms is to minimize the sum of squared distance between each point and its nearest centroid, i.e., the within-cluster sum of squares (WCSS) function.
Both DBSCAN and k-means (and their private versions) are unsupervised learning methods that work on unlabeled data. On the other hand, supervised learning methods like regression~\cite{DBLP:journals/popets/AlabiMSSV22}, support vector machine~\cite{senekane2019differentially} and deep learning~\cite{DBLP:conf/ccs/AbadiCGMMT016} have been studied under differential privacy, which all require labeled data.

There is also a line of research~\cite{DBLP:conf/asiaccs/BozdemirCEMO021,DBLP:journals/tdp/LiuXLH13,DBLP:conf/globecom/WangLCZH23,Fu2024PPADBSCANP,DBLP:journals/tifs/FuCSXCS24} known as "Privacy-Preserving Clustering," which leverages secure multiparty computation (MPC) techniques to ensure privacy. This approach is orthogonal to our work~\cite{DBLP:journals/popets/HegdeMSY21}:
differential privacy provides an information-theoretic privacy guarantee, but inherently introduces some error. In contrast, MPC protocols offer high accuracy but come with significant computational overhead, which scales with the number of participating parties.

\section{Preliminaries}
\label{sec:pre}


We call two datasets $P,P'\in\mathcal{P}$ neighbors, denoted $P\sim P'$, if they differ by a single point.
\begin{definition}[Differential Privacy~\cite{DBLP:journals/fttcs/DworkR14}]
A mechanism $\M:\mathcal{P}\to \mathcal{O}$ is $(\varepsilon,\delta)$-DP, if for all pairs of neighboring datasets $P\sim P'$ and any subset of outputs $O\subseteq \mathcal{O}$ it holds that
\[
\Pr[\M(P)\in O]\leq e^\varepsilon\cdot\Pr[\M(P')\in O]+\delta\,.
\]
When $\delta=0$, the mechanism is called pure-DP or $\varepsilon$-DP.
\end{definition}

\begin{theorem}[Group Privacy~\cite{DBLP:journals/fttcs/DworkR14}]\label{lemma:group}
We say $P$ and $P'$ are $k$-hop neighbors, if there is a sequence of datasets $P, P_1,\dots, P_{k-1}, P'$ of length $k$ where $P\sim P_1\sim\dots\sim P'$.
Given an $\varepsilon$-DP mechanism $\M$ and any subset of outputs $O\subseteq \mathcal{O}$, we have 
\[
\Pr[\M(P)\in O]\leq e^{k\varepsilon}\cdot\Pr[\M(P')\in O]\,.
\]
\end{theorem}


\begin{theorem}[Post-processing~\cite{DBLP:journals/fttcs/DworkR14}]
Let $\M$ be an $(\varepsilon,\delta)$-DP mechanism. For any randomized function $f$, $f\circ \M$ is also $(\varepsilon,\delta)$-DP.
\end{theorem}

Two flavors of differential privacy have been considered in the literature~\cite{DBLP:conf/sigmod/KiferM11}. In \textit{unbounded DP}, neighboring instances differ by adding or removing one point, and in \textit{bounded DP}, neighboring instances have the same size while differ by the value of exactly one point. It is easy to see that any pair of bounded DP neighbors are also 2-hop neighbors under unbounded DP. By group privacy, any $\varepsilon$-unbounded DP mechanism also satisfies $(2\varepsilon)$-bounded DP, or more formally:

\begin{theorem}\label{thm:bounded}
If there is an $\varepsilon$-unbounded DP mechanism that has error $\op{Err}(\varepsilon)$, then there is an $\varepsilon$-bounded DP mechanism that has error $\op{Err}(\frac{\varepsilon}{2})$. By contraposition, if there is no $\varepsilon$-bounded DP mechanism that has error $\op{Err}(\varepsilon)$, then there is no  $\varepsilon$-unbounded DP mechanism that has error $\op{Err}(2\varepsilon)$. 
\end{theorem}

In this paper, we adopt unbounded DP unless otherwise specified.
 
Next we introduce some standard DP mechanisms.

\begin{definition}
[$\ell_1$-sensitivity~\cite{DBLP:journals/fttcs/DworkR14}]
Given a function $f:\mathcal{P}\to\mathbb{R}^k$, its $\ell_1$-sensitivity is defined as $\Delta:=\max_{P\sim P'} \|f(P)-f(P')\|_1$.
\end{definition}

\begin{theorem}[Laplace Mechanism~\cite{DBLP:journals/fttcs/DworkR14}]
Given a function $f:\mathcal{P}\to\mathbb{R}^k$ that has $\ell_1$-sensitivity $\Delta$, the mechanism $\M(P)=f(P)+(Z_1, Z_2,\dots, Z_k)$ preserves $\varepsilon$-DP, where $Z_i$ are i.i.d.~ random variables drawn from $\op{Lap}(\Delta / \varepsilon)$ with $\Pr[Z=z]=\frac{\varepsilon}{2\Delta}e^{-\frac{\varepsilon|z|}{\Delta}}, z\in\mathbb{R}$.
\end{theorem}

\begin{theorem}[Geometric Mechanism~\cite{DBLP:journals/siamcomp/GhoshRS12}]\label{thm:geom}
Given a function $f:\mathcal{P}\to\mathbb{Z}^k$ that has $\ell_1$-sensitivity $\Delta$, the mechanism $\M(P)=f(P)+(Z_1, Z_2,\dots, Z_k)$ preserves $\varepsilon$-DP, where $Z_i$ are i.i.d.~ random variables drawn from $\op{Geom}(e^{\varepsilon /\Delta})$ with $\Pr[Z=z]=\frac{e^{\varepsilon/\Delta}-1}{e^{\varepsilon/\Delta}+1}e^{-\frac{\varepsilon|z|}{\Delta}}, z\in\mathbb{Z}$.
\end{theorem}


A histogram query $f(P)=(x_1,x_2,\dots,x_{|\mathcal{X}|})$ counts the number of points $x_i$ in the dataset for each $i\in\mathcal{X}$.
The $\ell_1$ sensitivity of $f$ is $\Delta=1$ (under unbounded DP, or $\Delta=2$ under bounded DP), thus the noise-adding mechanisms can naturally be applied.
We will use it as the building block of our mechanism. In fact, our algorithm can work with any DP histogram that has an error guarantee~\cite{DBLP:journals/vldb/XuZXYYW13,DBLP:conf/nips/Suresh19,DBLP:journals/jpc/BalcerV19,DBLP:conf/pods/LebedaT23}. For simplicity, we first focus on the standard Laplace histogram for presentation, and discuss the extensions in Section~\ref{sec:hist} and \ref{sec:discuss}.

\begin{theorem}[Laplace histogram]\label{thm:dphist}
Let $\mathcal{X}$ be a universe and $x_i$ be the count of points in dataset $P$ at $i\in\mathcal{X}$. An $\varepsilon$-DP Laplace histogram releases $\tilde{x}_i = x_i+\op{Lap}(\frac{1}{\varepsilon})$ for all $i\in \mathcal{X}$.
For any $0<\beta_0<1$, we have with probability $1-\beta_0$, its $\ell_\infty$ error is bounded by
\[
\Pr\left[\max_{i\in\mathcal{X}} |\tilde{x}_i - x_i| \geq \frac{1}{\varepsilon}\ln\frac{|\mathcal{X}|}{\beta_0}\right] \leq \beta_0\,.
\]
\end{theorem}

We will denote this error bound by $\gamma=O\left(\frac{1}{\varepsilon}\log\frac{|\mathcal{X}|}{\beta_0}\right)$.
When multiple independent Laplace or Geometric noises are summed, their error can be bounded using concentration inequalities.

\begin{theorem}[Laplace Concentration~\cite{DBLP:conf/icalp/ChanSS10}]\label{thm:lapConcentration}
Let $Z_1,\dots,Z_{\kappa}\sim \op{Lap}(\frac{1}{\varepsilon})$ be independent Laplace random variables, then for any $0<\beta_0<1$,
\[
\Pr\left[\left|\sum_{i=1}^\kappa Z_i\right|>\frac{2\sqrt{2}}{\varepsilon}\max\left\{\sqrt{\kappa\ln\frac{2}{\beta_0}},\ln\frac{2}{\beta_0}\right\}\right]\leq \beta_0\,.
\]
\end{theorem}

\begin{theorem}[Geometric Concentration~\cite{DBLP:conf/esa/ChanSS12}]\label{thm:geomConcentration}
Let $Z_1,\dots,Z_{\kappa}\sim \op{Geom}(e^{\varepsilon})$ be independent Geometric random variables, then for any $0<\beta_0<1$,
\[
\Pr\left[\left|\sum_{i=1}^\kappa Z_i\right|>\frac{4e^\varepsilon}{e^\varepsilon-1}\cdot\sqrt{\kappa}\cdot \ln\frac{2}{\beta_0}\right]\leq \beta_0\,.
\]
If in addition $\kappa\geq e^\varepsilon\cdot \ln\frac{2}{\beta_0}$, then
\[
\Pr\left[\left|\sum_{i=1}^\kappa Z_i\right|>\frac{4\sqrt{e^\varepsilon}}{e^\varepsilon-1}\sqrt{\kappa\ln\frac{2}{\beta_0}}\right]\leq \beta_0\,.
\]
\end{theorem}






\section{Approximate DBSCAN under DP}\label{sec:span}

In this section, we prove that the existing definition of DP-DBSCAN cannot achieve any meaningful utility guarantee, and propose the definition of spans.
Since any non-trivial DP mechanism must be randomized, for a mechanism that outputs cluster labels, we say it is $(\rho,\tau;\beta)$-accurate if it produces $(\rho,\tau)$-approximate clusters in Definition~\ref{def:sand} with probability at least $1-\beta$.
Smaller $\rho$, $\tau$ and $\beta$ gives better approximations.
In non-private setting, $\tau=\beta=0$ can be achieved for any constant $\rho>1$ in linear time \cite{DBLP:conf/sigmod/GanT15}.
But for DP-DBSCAN, we show that it is impossible for an algorithm to have a good utility with this definition.
We prove all the lower bounds in this paper under bounded DP. By Theorem~\ref{thm:bounded}, they also hold under unbounded DP up to a factor of $2$ in $\varepsilon$.

\begin{lemma}
Let $\M$ be any $\varepsilon$-bounded DP mechanism that is $(\rho,\tau;\beta)$-accurate for $(\alpha,\MinPts)$-DBSCAN, then $\beta\geq \frac{n}{n+e^\varepsilon}$. Correspondingly, any $\varepsilon$-unbounded DP mechanism that is $(\rho,\tau;\beta)$-accurate must have $\beta\geq \frac{n}{n+e^{2\varepsilon}}$. This holds for any $0<\tau<\MinPts$ and approximation ratio $\rho>0$.
\end{lemma}

\begin{proof}
We construct a 1D example.
Let $P$ be a dataset where all points are located at the origin, then DBSCAN should output a single cluster $\mathcal{C}(\alpha,\MinPts)=\{P\}$ for any $\alpha>0$ and $\MinPts\geq 1$. We fix $\alpha=\rho^{-1}$ and an arbitrary $\MinPts<n$. As $\M$ is $(\rho,\tau;\beta)$-accurate on $P$, the approximate cluster must be exactly $\{P\}$, thus $\Pr[\M(P)=\{P\}]\geq 1-\beta$. With high probability, all the points are assigned the same cluster id.

Consider a (bounded) neighbor $P_i'$, obtained by moving a single point $i$ from $p_i=0$ to $p_i'=1$.
There are still $n-1\geq \MinPts$ points at the origin and they belong to the same cluster. We have $P_i'-\{p_i'\}\subseteq \hat{C}$ for some $\hat{C}\in \M(P_i')$.
In addition, the second condition indicates $p_i'\not\in\hat{C}$ because otherwise we have $\dist(0,p_i')=1<\rho\alpha$, which is a contradiction. Let $O_i$ be all the outputs of $\M$ on $P_i'$ satisfying that 1) all points other than $i$ share the same non-zero label, and 2) point $i$ has a different label (which may be either a different cluster or a noise).
We then have $\Pr[\mathcal{M}(P_i')\in O_i]\geq 1-\beta$.

Note that $O_i$'s are disjoint for $i=1,\dots,n$, and by $\varepsilon$-DP, we have
\begin{align*}
\beta & \geq \Pr[M(P)\neq \{P\}] \\
& \geq \sum_{i=1}^n\Pr[M(P)\in O_i] \\
& \geq e^{-\varepsilon}\cdot \sum_{i=1}^n\Pr[M(P_i')\in O_i] \\
& \geq e^{-\varepsilon}n(1-\beta)
\end{align*}

Solve for $\beta$, we have $\beta\geq \frac{n}{n+e^\varepsilon}$. Plug in Theorem~\ref{thm:bounded}, we get the negative result for unbounded DP.
\end{proof}

Observe that for constant $\varepsilon$, the failure probability $\beta\approx 1$, and there is no meaningful utility guarantee.
Further, since the theorem holds for an arbitrarily large $\rho$, the hardness remains even when a large approximation ratio can be allowed, meaning no algorithm following the existing DP-DBSCAN definition can guarantee
good utility for all instances.
The main difficulty lies in the close connection between clustering results and the private points. In view of this, we propose a new way of describing a cluster $C$ called the \textit{span} of $C$, defined as $\Span(C)=\cup_{p\in C}\mathbf{B}(p,\alpha)=\{l\in \LL:\exists p\in C, d(p,l)<\alpha\}$.
Instead of points $p\in P$, a span is formed by locations $l\in \LL$, making it less sensitive to individual points, but as informative as the original clusters.

We argue that the definition of spans of clusters are well-justified with practical applications. In fact, noise points within the span of a cluster (or several clusters) are called \textit{boarder points} in~\cite{DBLP:conf/sigmod/GanT15}. They can be considered dense in a relaxed setting: for any $p'\in \Span(C)$, there is a core point $p\in C$ such that $\dist(p',p)<\alpha$, and therefore $|\mathbf{B}(p',2\alpha)\cap P|\geq |\mathbf{B}(p,\alpha)\cap P|\geq \MinPts$.
In practice, it is common for implementations of DBSCAN to include visualizations more than scatter points, e.g.~convex hulls~\cite{JSSv091i01}. Our definition of spans can achieve the same functionality in a more fine-grained way since DBSCAN supports finding non-convex clusters (e.g.~Figure~\ref{fig:spanCirc}).
Further, spans can offer a geometric interpretation of the shapes and locations of clusters, which cannot be done by existing DP-DBSCAN mechanisms that outputs cluster labels while hiding point locations.

\begin{remark}
Our definition of spans of clusters can be viewed as a DBSCAN analogy to the centers under k-means clustering, or equivalently, the Voronoi diagram formed by these centers.
Note that DP Kmeans protects privacy by outputting the centers only, without revealing individual labels. A person can compare his/her data with all the centers to find the nearest center as his/her cluster. This is essentially querying the Voronoi diagram formed by the privatized centers. The spans of clusters offer exactly the same functionality for DBSCAN.
\end{remark}

Similarly, for span-based mechanisms, we say it is $(\rho,\tau;\beta)$-accurate if it produces $(\rho,\tau)$-approximate spans in Definition~\ref{def:ratio} with probability at least $1-\beta$.
Compared with Definition~\ref{def:sand} from~\cite{DBLP:conf/sigmod/GanT15}, the first condition is unchanged.
The second condition is naturally relaxed from $\hat{C}\subseteq C_2$ to $S\subseteq\Span(C_2)$, which bounds how large the approximate spans can be.

\section{Lower Bounds}\label{sec:lower}

In this section, we show that even when we consider span-based mechanisms, both approximation factors $\rho$ and $\tau$ are still necessary.
We first show that the approximation on $\MinPts$ is $\tau=\Omega(\frac{1}{\varepsilon}\log \frac{1}{\rho\alpha})$ using a packing argument.

\begin{theorem}
For $\rho>1$, any $\varepsilon$-bounded or unbounded DP mechanism $\M$ that is $(\rho,\tau;0.1)$-accurate for $(\alpha,\MinPts)$-clustering must have $\tau=\Omega(\frac{1}{\varepsilon}\log \frac{1}{\rho\alpha})$.
\end{theorem}

\begin{proof}
We introduce a notation for simplicity in the proof. Let $\mathcal{S}$ be any output of $\M$. We say $\mathcal{S}$ covers location $l$, denoted $\mathcal{S}\vdash l$, if $l$ is covered by at least one approximate span in $\mathcal{S}$. Namely, $\mathcal{S}\vdash l$ if and only if $\exists S\in \mathcal{S}, l\in S$. Correspondingly, $\mathcal{S}\nvdash l$ if $\forall S\in\mathcal{S}, l\not\in S$.

We construct an example in 1D, with DBSCAN parameters $\alpha$ and $\MinPts=\tau+2$.
Consider $K$ locations in $[0, 1]$, which we label by $l_k=2\rho\alpha \cdot k$ for $k=1,\dots,K$ where $K=\lfloor\frac{1}{2\rho\alpha}\rfloor$. Note that the distance between any pair of locations is $\dist(l_k, l_j)>\rho\alpha$ for $k\neq j$.
We will consider $K$ point sets $P_1,\dots,P_K$.
Each dataset $P_k$ contains $n=\tau+K+1$ points, with $\tau+2$ points located at $l_k$ and a single point at each other $l_j$, $j\neq k$.
Then for any $j\neq k$, $P_j$ and $P_k$ are $(\tau+1)$-hop (bounded) neighbors.
By group privacy (Lemma~\ref{lemma:group}), this means for any set of outputs $O$,
$\Pr[\M(P_k)\in O]\leq e^{(\tau+1)\varepsilon}\cdot \Pr[\M(P_j)\in O]$.

We then classify all possible outputs of $\M$ into sets $O_i$.
In particular, we denote $\mathcal{S}\in O_k$ if the output $\mathcal{S}$ covers location $l_k$ but not any other location $l_j$ for $j\neq k$. That is, 
\[
O_k=\{\mathcal{S}: \mathcal{S}\vdash l_k\text{ and }\forall j\neq k, \mathcal{S}\nvdash l_j\}\,.
\]
We argue that if $\M$ is $(\rho,\tau;0.1)$-accurate, then with high probability, running $\M$ on $P_k$ will produce some $\mathcal{S}$ that belongs to $O_k$, which gives $\Pr[\M(P_k)\in O_k]\geq 0.9$.
To see this, first observe that there is a unique $(\alpha,\tau+2)$-cluster on $P_k$, containing all the points at location $l_k$.
Even with relaxed parameters $(\rho\alpha, 2)$, there is still the same unique cluster, since any point not at $l_k$ have only a density of $1<2$ in its neighborhood.
Then by the first guarantee in Definition~\ref{def:ratio}, there exists  $S\in \mathcal{S}$ containing points at $l_k$, meaning $\mathcal{S} \vdash l_k$.
By the second guarantee, any $S\in \mathcal{S}$ must satisfy $S\subseteq \Span(C_2)=\mathbf{B}(l_k,\rho\alpha)$.
But we have $\dist(l_k, l_j)>\rho\alpha$ for any $j\neq k$, it follows that $l_j\not\in S$ for any $S\in\mathcal{S}$, that is $\mathcal{S}\nvdash l_j$ for any $j\neq k$.
Combine both, we have $\mathcal{S}\in O_k$, which happens with probability at least $0.9$.

Finally note that the sets $O_k$'s are disjoint. So we conclude
\begin{align*}
0.1&\geq \Pr[\M(P_1)\not\in O_1] \\
&\geq \sum_{k=2}^K \Pr[\M(P_1)\in O_k] \\
&\geq \sum_{k=2}^K e^{-(\tau+1)\varepsilon}\Pr[\M(P_k)\in O_k] \\
&\geq 0.9\cdot \left(\frac{1}{2\rho\alpha}-2\right)\cdot e^{-(\tau+1)\varepsilon}\,,
\end{align*}
which gives $\tau=\Omega(\frac{1}{\varepsilon}\log \frac{1}{\rho\alpha})$. By Theorem~\ref{thm:bounded}, we have the same asymptotic result for unbounded DP.
\end{proof}

We next show that even after relaxing to $\MinPts-\tau$, the relaxation from $\alpha$ to $\rho\alpha$ is still necessary.
In particular, $\rho\geq 3$ for any $\tau\leq n/2-1$.

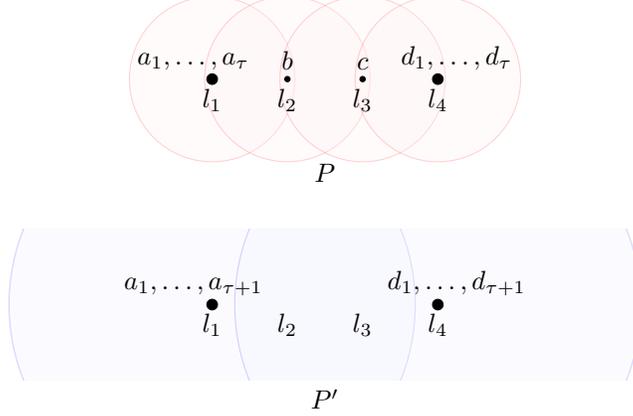
\begin{figure}[htbp]
\centering
\begin{tikzpicture}

\filldraw[color=red!60, fill=red!5,opacity=0.3](0,3) circle (1.1);
\filldraw[color=red!60, fill=red!5,opacity=0.3](1,3) circle (1.1);
\filldraw[color=red!60, fill=red!5,opacity=0.3](2,3) circle (1.1);
\filldraw[color=red!60, fill=red!5,opacity=0.3](3,3) circle (1.1);

\draw (-0.25,3) node[anchor=south] {$a_1,\dots,a_\tau$};
\draw[fill] (0,3) circle (2pt) node[anchor=north] {$l_1$};

\draw (1,3) node[anchor=south] {$b$};
\draw[fill] (1,3) circle (1pt) node[anchor=north] {$l_2$};

\draw (2,3) node[anchor=south] {$c$};
\draw[fill] (2,3) circle (1pt) node[anchor=north] {$l_3$};

\draw (3.25,3) node[anchor=south] {$d_1,\dots,d_\tau$};
\draw[fill] (3,3) circle (2pt) node[anchor=north] {$l_4$};
\draw (1.5,2) node[anchor=north] {$P$};

\begin{scope}
\clip (-3,-1) rectangle (6,1);
\filldraw[color=blue!60, fill=blue!5,opacity=0.3](0,0) circle (2.7);
\filldraw[color=blue!60, fill=blue!5,opacity=0.3](3,0) circle (2.7);
\end{scope}

\draw (-0.25,0) node[anchor=south] {$a_1,\dots,a_{\tau+1}$};
\draw[fill] (0,0) circle (2pt) node[anchor=north] {$l_1$};
\draw (1,0) node[anchor=north] {$l_2$};
\draw (2,0) node[anchor=north] {$l_3$};
\draw (3.25,0) node[anchor=south] {$d_1,\dots,d_{\tau+1}$};
\draw[fill] (3,0) circle (2pt) node[anchor=north] {$l_4$};
\draw (1.5,-1) node[anchor=north] {$P'$};
\end{tikzpicture}
\caption{Proof for Theorem~\ref{thm:lbRho}.}\label{fig:proof}
\end{figure}

\begin{theorem}\label{thm:lbRho}
For $\varepsilon\leq 1$ and $\tau\leq n/2-1$, any $\varepsilon$-bounded or unbounded DP mechanism $\M$ that is $(\rho,\tau; 0.1)$-accurate must have $\rho\geq 3$.
\end{theorem}

\begin{proof}
With out loss of generality assume $n$ is even.
We only need to construct an example where $\tau=n/2-1$, since more accurate $\tau$ makes the problem harder for $\rho$.
We construct an example in 1D, with parameters $\alpha$ and $\MinPts=\tau+1$.
Assume such an $\mathcal{M}$ exists for some $\rho<3$, we consider 4 locations $l_k=r\cdot k$ for $k=1,2,3,4$ where $r=\rho\alpha/3<\alpha$.
We consider two datasets $P$ and $P'$ in Figure~\ref{fig:proof}, defined as 
$P=(a_1,\dots,a_\tau, b, c, d_1,\dots,d_\tau)$, where $a_i=l_1$, $b=l_2$, $c=l_3$ and $d_i=l_4$ for $i=1,\dots,\tau$; and 
$P'=(a_1,\dots,a_\tau, d_1,\dots,d_\tau)$ is a 2-hop (bounded) neighbor of $P$ where the locations for $b$ and $c$ are different. It is easy to see that all the points in $P$ are core points in $(\alpha,\tau+1)$-clustering. Moreover, they all belong to the same cluster since $\dist(l_1,l_2)=\dist(l_2,l_3)=\dist(l_3,l_4)=r<\alpha$.
Let $O$ be the set of outputs that put these 4 locations into the same span, namely
\[
O=\{\mathcal{S}: \exists S\in\mathcal{S}\text{ such that } \{l_1,l_2,l_3,l_4\}\subseteq S\}\,.
\]
We have $\Pr[\M(P)\in O]\geq 0.9$ when $\M$ is $(\rho,\tau;0.1)$-accurate.

Since $\M$ is $\varepsilon$-DP, we have $\Pr[\M(P')\in O]\geq 0.9\cdot e^{-2\varepsilon}>0.1$ due to group privacy.
By the accuracy on $P'$, this means some output in $O$ must also satisfy Definition~\ref{def:ratio}.
In particular, there exists a cluster $C_2\in \mathcal{C}(\rho\alpha, 1)$ such that $\{l_1,l_2,l_3,l_4\}\subseteq S\subseteq \Span(C_2)$.
However, note that on $P'$, the output of DBSCAN is $\mathcal{C}(\rho\alpha, 1)=\{\{l_1\}, \{l_4\}\}$ as $\dist(l_1,l_4)=3r=\rho\alpha$.
When $\rho<3$, it is impossible for the open balls $\Span(\{l_1\})=\mathbf{B}(l_1,\rho\alpha)$ or $\Span(\{l_4\})=\mathbf{B}(l_4,\rho\alpha)$ to cover all the 4 locations.
We therefore conclude $\rho\geq 3$ by contradiction.
\end{proof}




\section{DP-DBSCAN on Noisy Histogram}\label{sec:upper}

In this section, we provide a $(3+\eta, \tau; \beta)$-accurate mechanism $\M$ that is $\varepsilon$-DP, where $\eta$ is a small constant and 
\[
\tau = O\left(\left(1+\frac{8\sqrt{d}}{\eta}\right)^d\cdot \frac{d}{\varepsilon}\log\frac{d}{\alpha\beta}\right)\,.
\]
For constant $d$ and $\beta$, we have $\tau=O(\frac{1}{\varepsilon}\log\frac{1}{\alpha})$, so that both approximation factors $\rho$ and $\tau$ match our lower bound of $(3, \Omega(\frac{1}{\varepsilon}\log\frac{1}{\alpha}); 0.1)$ in Section~\ref{sec:lower}, up to constant factors. 

Inspired by~\cite{DBLP:conf/sigmod/GanT15}, we first divide the space into a grid of cells with width $w=\eta' \cdot \alpha/\sqrt{d}$, where $0< \eta' \leq 1$ is a small constant that controls the size of the cells. Any pair of points in the same cell has mutual distance less than $\alpha$, so they must belong to the same cluster if they are core points.
Therefore, we can perform clustering on the cells to avoid releasing information of a specific user.
We call a cell \textbf{core cell} if it contains at least one core point, otherwise it is a \textbf{noise cell} which may be empty or containing noise points.

Our algorithm works by first finding a superset of the core cells, followed by merging cells that belong to the same cluster. We present the full algorithm in Algorithm~\ref{algo} and illustrate each step in the following subsections.

\begin{algorithm}[htbp]
\begin{algorithmic}[1]
\Require Private dataset $P$, DBSCAN parameters $(\alpha,\MinPts)$, constant  $0<\eta'\leq 1$, privacy budget $\varepsilon$, failure probability $\beta$
\Ensure Approximate spans $\mathcal{S}=\{S_1,S_2,\dots\}$
\State $w\gets \eta'\cdot\alpha/\sqrt{d}$
\State Divide $\LL$ into a grid of cells $\X=\{X_1,\dots,X_{|\X|}\}$ of width $w$
\State Release a DP histogram for the noisy count $\tilde{x}_i$ for $X_i\in\mathcal{X}$
\State Find a superset of core cells $\X_{core}:=\{X: \widetilde{\op{UB}}(X) \geq \MinPts\}$
\State Merge cells in $\X_{core}$ to form approximate spans $\mathcal{S}$
\end{algorithmic}
\caption{DP Approximate DBSCAN}\label{algo}
\end{algorithm}

\subsection{DP Histogram on Cells}

Let the cells be $\X=\{X_1,X_2,\dots,X_{|\X|}\}$, where $|\X|=w^{-d}=\left(\frac{\sqrt{d}}{ \eta' \alpha}\right)^d$.
We first run a DP histogram algorithm to release the number of points  $x_i := |X_i\cap P|$ in every cell $X_i\in \X$.
By Theorem~\ref{thm:dphist}, the maximum error of any $\tilde{x}_i$ is
$\gamma=O(\frac{1}{\varepsilon}\log\frac{|\X|}{\beta_0})=O(\frac{d}{\varepsilon}\log\frac{d}{\alpha\beta_0})$.

Note that by post-processing, the privacy guarantee of our mechanism is obvious. This also implies that any DP histogram with an accuracy guarantee $\gamma$ can be used.
Further, our post-processing steps are all deterministic, and the success of our algorithm only depends on the histogram being accurate. Therefore we will use $\beta_0=\beta$ in the remaining of this section.

\subsection{Superset of Core Cells}

In this subsection, we find a superset of the core cells to ensure all the real clusters are reported, so that the first condition in Definition~\ref{def:ratio} is satisfied.

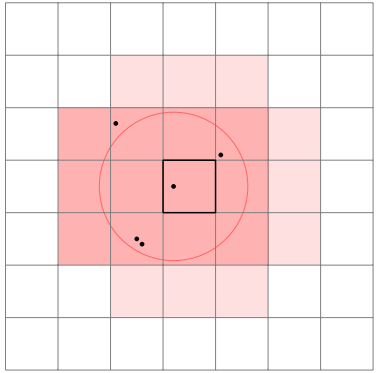
\begin{figure}[htbp]
\centering
\resizebox{0.3\textwidth}{!}{
\begin{tikzpicture}
\fill[color=red!60, opacity=0.5] (1,2) rectangle (5,5);
\fill[color=red!60, opacity=0.2] (2,1) rectangle (5,2);
\fill[color=red!60, opacity=0.2] (2,5) rectangle (5,6);
\fill[color=red!60, opacity=0.2] (5,2) rectangle (6,5);
\draw[color=red!60](3.2,3.5) circle (1.414);
\draw[fill] (3.2,3.5) circle (1pt);
\draw[fill] (2.5,2.5) circle (1pt);
\draw[fill] (2.6,2.4) circle (1pt);
\draw[fill] (2.1,4.7) circle (1pt);
\draw[fill] (4.1,4.1) circle (1pt);

\draw[help lines] (0,0) grid (7,7);
\draw[thick] (3,3) grid (4,4);
\end{tikzpicture}}
\caption{$\alpha$-neighborhood of a 2D cell ($\eta'=1$).}\label{fig:cell}
\end{figure}

Let $X$ be any cell.
Observe that when $\eta'$ is constant, the union of $\alpha$-neighborhoods of its points $\cup_{p\in X} \mathbf{B}(p,\alpha)$ only intersects a constant number of cells.
A 2D example from~\cite{DBLP:conf/sigmod/GanT15} is shown in Figure~\ref{fig:cell}, where when $\eta'=1$, the 2D circles from points in $X$ can only intersect 21 cells, including $X$ itself.
We call these cells the $\alpha$-neighborhood of $X$, denoted $\mathbf{NB}(X)$.
There are at most $\kappa=|\mathbf{NB}(X)|\leq \left(\frac{2\alpha+w}{w}\right)^d=\left(1+\frac{2\sqrt{d}}{\eta'}\right)^d=O(1)$ neighbor cells for constant $d$ and $\eta'$.

Let $\op{UB}(X)$ be the sum of counts in the $\alpha$-neighborhood of cell $X$, it follows that for any $p\in X$,
\begin{equation}\label{eq:ub}
\op{UB}(X):=\sum_{X_i\in \mathbf{NB}(X)} |X_i\cap P| \geq |\mathbf{B}(p,\alpha) \cap P|\,.
\end{equation}
By reporting all cells $X$ such that $\op{UB}(X)\geq \MinPts$, we are guaranteed to find all the core cells, with some possible false positives. Now with the noisy histogram, we cannot obtain $\op{UB}(X)$ in exact. Instead, our algorithm uses a noisy upper bound $\widetilde{\op{UB}}(X)$ defined as
\begin{equation}\label{eq:ubesti}
\widetilde{\op{UB}}(X)=\left(\sum_{X_i\in \mathbf{NB}(X)} \tilde{x}_i\right) + \Gamma\,,
\end{equation}
where $\tilde{x}_i$ is the noisy count of cell $X_i$ and $\Gamma(\kappa,\gamma)$ is a data-independent upward-scaling factor that ensures $\widetilde{\op{UB}}(X)\geq \op{UB}(X)$. For now, simply take $\Gamma = \kappa\gamma = O\left(\left(1+\frac{2\sqrt{d}}{\eta'}\right)^d\cdot \frac{d}{\varepsilon}\log\frac{d}{\alpha\beta}\right)$, so that conditioned on a histogram with error bound by $\gamma$, the sum of $\kappa$ noises will be bounded by $\Gamma$.
We conclude the guarantee for the core cells in the following theorem.

\begin{lemma}\label{thm:core:cell}
Let $\X_{core}=\{X: \widetilde{\op{UB}}(X)\geq \MinPts\}$, and let $\tau =2 \Gamma = O\left(\left(1+\frac{2\sqrt{d}}{\eta'}\right)^d\cdot \frac{d}{\varepsilon}\log\frac{d}{\alpha\beta}\right)$.
Conditioned on the DP histogram being accurate, which happens with probability $1-\beta$, we have
\begin{enumerate}
\item For any core cell $X$ that contains an $(\alpha,\MinPts)$-clustering core point $p\in X$, we have $X\in \X_{core}$.
\item When $\MinPts>\tau$, for any $X\in\X_{core}$, there exists a core point $p$ in $((2+3\eta')\alpha, \MinPts-\tau)$-clustering, such that $X$ lies within the ball $X\subseteq \mathbf{B}(p,(1+2\eta')\alpha)$.
\end{enumerate}
\end{lemma}

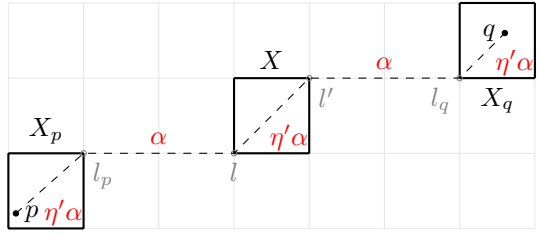
\begin{figure}[htbp]
\centering
\begin{tikzpicture}
\draw[help lines, opacity=0.2] (0,0) grid (7,3);
\draw[thick] (3,1) grid (4,2);
\draw[] (3.5,2) node[anchor=south]{$X$};
\draw[thick] (0,0) grid (1,1);
\draw[] (0.5,1) node[anchor=south]{$X_p$};
\draw[thick] (6,2) grid (7,3);
\draw[] (6.5,2) node[anchor=north]{$X_q$};

\draw[fill] (0.1,0.2) circle (1pt) node[anchor=west]{$p$};
\draw[fill] (6.6,2.6) circle (1pt) node[anchor=east]{$q$};
\draw[gray] (3,1) circle (1pt) node[anchor=north]{$l$};
\draw[gray] (4,2) circle (1pt) node[anchor=north west]{$l'$};
\draw[gray] (1,1) circle (1pt) node[anchor=north west]{$l_p$};
\draw[gray] (6,2) circle (1pt) node[anchor=north east]{$l_q$};

\draw[red] (2,1) node[anchor=south]{$\alpha$};
\draw[red] (5,2) node[anchor=south]{$\alpha$};
\draw[red] (3.35,1.5) node[anchor=north west]{$\eta'\alpha$};
\draw[red] (0.35,0.5) node[anchor=north west]{$\eta'\alpha$};
\draw[red] (6.35,2.5) node[anchor=north west]{$\eta'\alpha$};

\draw[dashed] (0.1,0.2) -- (1,1) -- (3,1) -- (4,2)
   -- (6,2) -- (6.6,2.6);

\end{tikzpicture}
\caption{Proof for $(2+3\eta')$-approximation.}\label{fig:ratio}
\end{figure}

\begin{proof}
The accuracy of DP histogram gives $x_i-\gamma\leq \tilde{x}_i\leq x_i+\gamma$ simultaneously for all cells. Conditioned on this, we only need to note that for any cell $X\ni p$ containing a core point, we have
\[
\widetilde{\op{UB}}(X)\geq \op{UB}(X) \geq |\mathbf{B}(p,\alpha) \cap P| \geq  \MinPts\,,
\]
and the first property is proved.

For the second, note that for any $X\in\X_{core}$, 
\[
\MinPts\leq \widetilde{\op{UB}}(X)\leq \op{UB}(X) +2\Gamma=\op{UB}(X)+\tau\,.
\]
Thus, the counted cells in $\mathbf{NB}(X)$ contains at least $\MinPts-\tau>0$ points.
We fix an arbitrary point in the neighboring cells $p\in X_p$ where $X_p\in \mathbf{NB}(X)$, and argue that $\dist(p,q)<(2+3\eta')\alpha$ for any other $q\in X_q$ where $X_q\in \mathbf{NB}(X)$, to show that $p$ is a core-point in $((2+3\eta')\alpha, \MinPts-\tau)$-clustering.

As shown in Figure~\ref{fig:ratio}, since $X_p\in \mathbf{NB}(X)$, there must exist a pair of locations $l\in X$ and $l_p\in X_p$, so that $\dist(l,l_p)<\alpha$. Similarly, there exists $l'\in X$ and $l_q\in X_q$ so that $\dist(l',l_q)<\alpha$. Finally, note that any two locations in the same cell
have distance at most $w\sqrt{d}=\eta'\alpha$. We then have
\begin{align*}
\dist(p,q)&\leq \dist(p,l_p)+\dist(l_p,l)+\dist(l,l')+\dist(l',l_q)+\dist(l_q, q) \\
&<(2+3\eta')\alpha\,.
\end{align*}
As this holds for any $q$ out of cells $\mathbf{NB}(X)$ and there are $\MinPts-\tau$ points within (including $p$), we conclude that $p$ is a core point in $((2+3\eta')\alpha, \MinPts-\tau)$-clustering. Further, for any $l''\in X$, we have $\dist(p,l'')\leq \dist(p,l_p) + \dist(l_p, l) + \dist(l, l'') < (1+2\eta')\alpha$. So we have $X\subseteq \mathbf{B}(p, (1+2\eta')\alpha)$.
\end{proof}

\subsection{Merge Cells into Clusters}

Given the core cells found above, we merge them to form approximate spans $\mathcal{S}=\{S_1,S_2,\dots\}$, each being a set of locations formed by a union of cells.
To guarantee that two core points of the same cluster belong to the same span, our merging policy is to merge a pair of cells if their minimum distance is less than $\alpha$, namely we merge $X\in\X_{core}$ and $X'\in\X_{core}$ if  $\dist(X,X'):=\min_{l\in X, l'\in X'} \dist(l,l')<\alpha$.
This brings the main theorem of our paper.

\begin{theorem}[Asymptotic Bound]\label{thm:main}
For $0<\eta \leq 4$, there is an $\varepsilon$-DP mechanism for $(\alpha,\MinPts)$-DBSCAN clustering that is $(3+\eta,\tau;\beta)$-accurate, where $\tau = O\left(\left(1+\frac{8\sqrt{d}}{\eta}\right)^d\cdot \frac{d}{\varepsilon}\log\frac{d}{\alpha\beta}\right)$, when $\MinPts>\tau$.
\end{theorem}

\begin{proof}
We prove a $(3+4\eta',\tau; \beta)$-approximation, so that taking $\eta=4\eta'$ gives the result claimed.

For the first condition, consider any $(\alpha,\MinPts)$-cluster $C_1$. In Lemma~\ref{thm:core:cell} we have shown that our algorithm reports all the core cells. By definition, for any pair of core points $p,q\in C_1$, there is a sequence of core points $p_1,\dots,p_k\in C$ where $p_1=p$,  $p_k=q$ and $\dist(p_i,p_{i+1})<\alpha$.
As a result, their corresponding cells $X_1,\dots,X_k$ satisfies $\dist(X_i, X_{i+1})<\alpha$, and will be merged into $S$ so that $C_1\subseteq S$.

For the second condition, although the cells are all from spans of $((2+3\eta')\alpha, \MinPts-\tau)$-clusters, our lower bound indicates it is impossible to have a $(2+3\eta', \tau)$-approximation.
Instead, we argue that whenever we merge cells, they are from the span of the same $((3+4\eta')\alpha, \MinPts-\tau)$ cluster.

Consider when we merge $X$ and $X'$, and the  distance is minimized at $\dist(l,l')<\alpha$ where $l\in X$ and $l'\in X'$.
By Theorem~\ref{thm:core:cell}, $X\subseteq \mathbf{B}(p,(1+2\eta')\alpha)$, where $p$ is a $((2+3\eta')\alpha, \MinPts-\tau)$-core point. Similarly, $X'\subseteq \mathbf{B}(p',(1+2\eta')\alpha)$ for core point $p'$.
We then have $\dist(p,p')\leq \dist(p,l)+\dist(l,l')+\dist(l',p')< (3+4\eta')\alpha$.
So $p$ and $p'$ are core-points from the same $((3+4\eta')\alpha, \MinPts-\tau)$ cluster.
\end{proof}

\section{A Linear-Time Pure-DP Histogram}\label{sec:hist}

The space and time of our mechanism depend on the complexity of the noisy histogram we use. After the histogram is built, identifying the core cells requires only a linear scan over the non-zero entries in the histogram, and merging neighboring cells requires comparing each non-zero entry with its $\kappa=O(1)$ neighbors. A naive implementation of the Laplace histogram in Theorem~\ref{thm:dphist} by adding independent noise to each cell will consume $O(|\X|)=O(w^{-d})$ time and space, which can be much larger than $n$, the data size.

As mentioned, our algorithm accepts any DP histogram with an explicit error guarantee $\gamma$ (or $\Gamma$). For approximate-DP, the stability-based histogram~\cite{DBLP:journals/jpc/BalcerV19} can be built in\footnote{
$\tilde{O}(\cdot)$ suppresses powers of $1/\varepsilon$, $\log 1/\delta$, $\log n$ and $\log |\X|$.
} $\tilde{O}(n)$ time using $O(n)$ space.
Under pure-DP, \cite{DBLP:conf/icdt/CormodePST12} designed a high-pass filter algorithm to only output the frequencies above a threshold in the DP histogram. The algorithm in \cite{DBLP:conf/icdt/CormodePST12} adopts Geometric noises; we use a version that uses the Laplace distribution, which enjoys a tighter concentration  bound.  Algorithm~\ref{algo:hist} takes as input the universe $\X$, a frequency map $F$ that contains the frequency $x_i>0$ for each element in the dataset, the privacy budget $\varepsilon$, and a threshold $\theta$ that decides which noisy frequencies will be truncated. In line \ref{line:freq} to \ref{line:freqend}, Laplace noise is added to each non-zero frequency, and the noisy $\tilde{x}_i$ is included into the output if it is at least $\theta$. Next, the distribution for the remaining $M$ entries of $\X-F$ is simulated by first sampling a set of entries that will have positive $\hat{x}_j$, and then drawing their noisy counts from the upper tail of the Laplace distribution.

\begin{algorithm}[htbp]
\begin{algorithmic}[1]

\Require Universe $\X$, frequency map $F=\{i\to x_i\}$,  privacy budget $\varepsilon$, threshold $\theta$
\Ensure Noisy histogram $H=\{j\to  h_j\}$ 

\State $H\gets \emptyset$
\ForEach {$(i,x_i)\in F$}\label{line:freq}
    \State $\tilde{x}_i\gets x_i + \op{Lap}(\frac{1}{\varepsilon})$
    \If {$\tilde{x}_i\geq \theta$}
        \State $H\gets H\cup \{i \to \tilde{x}_i\}$\label{line:freqend}
    \EndIf
\EndFor
\State $M\gets|\X-F|$, $p\gets \frac{1}{2}e^{-\varepsilon \theta}$
\State Sample $m\sim \op{Bin}(M, p)$ \label{line:bin}
\State Sample $m$ elements without replacement $J \subseteq \X-F$\label{line:srs}
\ForEach {$j\in J$}
    \State Sample $\hat{x}_j\geq \theta$ with $\Pr[\hat{x}_j=z]=\frac{1}{p} \cdot \Pr[\op{Lap}(\frac{1}{\varepsilon})=z]$
    \State $H\gets H\cup \{j \to \hat{x}_j\}$
\EndFor
\State \textbf{Output} $H$
\end{algorithmic}
\caption{Linear Time DP Histogram}\label{algo:hist}
\end{algorithm}

In \cite{DBLP:conf/icdt/CormodePST12}, it was only shown that the distribution of each $\hat{x}_i$ is equivalent to that in a standard histogram, without showing that they are independent, which is needed to show that the entire histogram satisfies DP.  Below we provide a complete proof on both privacy and utility. 

\begin{theorem}
For any $\theta>0$, Algorithm~\ref{algo:hist} is $\varepsilon$-DP. For any $\theta \geq \frac{1}{\varepsilon}\ln\frac{|\X|}{n}$, it runs in $O(n)$ time with high probability (in $n$). For any $0<\beta_0<1$, we have with probability $1-\beta_0$, the $\ell_\infty$ error of the output histogram $H$ is bounded by $\Pr[\max_{i\in\X} |h_i -x_i|>\theta+\gamma]\leq \beta_0$, where $\gamma=O(\frac{1}{\varepsilon}\log\frac{|\X|}{\beta_0})$ is the error for Laplace histogram.
\end{theorem}
\begin{proof}
Without loss of generality, let $\X=\{1,\dots,|\X|\}$ where  $x_1,\dots,x_M$ are all $0$'s and $x_{M+1},\dots,x_{|\mathcal{X}|}>0$ in $F$.
We show that the distribution of $H=(h_1,\dots,h_{|\X|})$ is equivalent to that of $Y=(y_1,\dots,y_{|\X|})$ where $y_i = \tilde{x}_i\cdot \mathbf{1}[\tilde{x}_i\geq\theta]$ is obtained by removing noisy counts less than $\theta$ from the Laplace histogram in Theorem~\ref{thm:dphist}. Note that\footnote{
As is standard in DP analysis, we use $\Pr[\cdot]$ to denote the probability density function of continuous distributions.
}
\[
\Pr[y_i=z \mid x_i]=\begin{cases}
\Pr[x_i + \op{Lap}(\frac{1}{\varepsilon})<\theta]\,, & z = 0\,, \\
\Pr[x_i + \op{Lap}(\frac{1}{\varepsilon})=z]\,, 
 & z\geq\theta\,, \\
0\,, & \text{otherwise} \,. \\
\end{cases}
\]
For the zero-entries, $y_1,\dots, y_M$ are i.i.d., abbreviated as
\begin{equation}\label{eq:r}
\Pr[y=z]=\begin{cases}
\Pr[\op{Lap}(\frac{1}{\varepsilon})<\theta]=1-p\,, & z = 0\,, \\
\Pr[\op{Lap}(\frac{1}{\varepsilon})=z]\,, 
 & z\geq\theta\,, \\
0\,, & \text{otherwise} \,, \\
\end{cases}
\end{equation}
where $p=\frac{1}{2}e^{-\varepsilon \theta}$ for Laplace distribution.
By independence, the joint probability distribution of $Y$ is given by
\begin{equation}\label{eq:Y_distribution}
\Pr\left[Y=(z_1,\dots,z_{|\X|})\right] = \prod_{j=1}^M \Pr[y=z_j] \cdot \prod_{i=M+1}^{|\X|} \Pr[y_i=z_i \mid x_i]\,.
\end{equation}

Now consider $H$, note that when $x_i>0$, $h_i$ is generated using the exact same procedure as $y_i$.
For $x_j=0$, conditioned on the simple random sample (SRS) $J \subseteq\{1,\dots M\}$, each $h_j$ is distributed according to either $h_j\equiv 0$ (for $j\not\in J$) or according to
\[
\Pr[h=z]=\begin{cases}
 0\,, & z <\theta \,, \\
\frac{1}{p} \cdot \Pr[\op{Lap}(\frac{1}{\varepsilon})=z]\,, 
 & z\geq\theta\,, \\
\end{cases}\quad (\text{for } j\in J)\,.
\]

Now fix an arbitrary output $\z=(z_1,\dots,z_{|\X|})$, $z_i\in \{0\}\cup [\theta, \infty)$. Let $J_{\z}\subseteq\{1,\dots M\}$ be the indices of non-zero values in the first $M$ entries of $\z$ and $|J_{\z}|=m_{\z}$, it is easy to see that the algorithm produces $\z$ only if $m_{\z}$ and $J_{\z}$ are sampled. By conditional probability,
\begin{align*}
&\; \Pr\left[H=\z\right] \\
=&\;\Pr[\op{Bin}(M, p)=m_{\z}] \cdot \Pr[\op{SRS}(M, m_{\z})=J_{\z} \mid m_{\z}] \\
&\; \cdot \prod_{j=1}^M \Pr[h_j = z_i \mid J_{\z}] \cdot \prod_{i=M+1}^{|\X|} \Pr[y_i=z_i \mid x_i] \\
=&\; \binom{M}{m_{\z}} p^{m_{\z}} (1-p)^{M-m_{\z}} \cdot \frac{1}{\binom{M}{m_{\z}}} \cdot 1^{M-m_{\z}} \\
 &\; \cdot \frac{1}{p^{m_{\z}}}\cdot \prod_{j\in J_{\z}} \Pr[\op{Lap}(\frac{1}{\varepsilon}) = z_j] \cdot \prod_{i=M+1}^{|\X|} \Pr[y_i=z_i \mid x_i] \\
 = &\; (1-p)^{M-m_{\z}} \cdot \prod_{j\in J_{\z}} \Pr[\op{Lap}(\frac{1}{\varepsilon}) = z_j] \cdot \prod_{i=M+1}^{|\X|} \Pr[y_i=z_i \mid x_i] \\
\end{align*}

On the other hand, plug Equation~\eqref{eq:r} into Equation~\eqref{eq:Y_distribution} for the same $\z$,
\[
\Pr\left[Y=\z\right] = (1-p)^{M-m_{\z}}\cdot \prod_{j\in J_{\z}} \Pr[\op{Lap}(\frac{1}{\varepsilon}) = z_j] \cdot \prod_{i=M+1}^{|\X|} \Pr[y_i=z_i|x_i]\,.
\]
Therefore $\Pr[H=\z]=\Pr[Y=\z]$ for all $\z$, proving they are identically distributed.

\textbf{Privacy.} Privacy follows from the simple fact that $Y$ is obtained by post-processing the results of a Laplace histogram which satisfies $\varepsilon$-DP, and that $H$ is identically distributed as $Y$.

\textbf{Accuracy.} 
Consider any frequency $x_i$, let $\tilde{x}_i$ be the noisy frequency in a Laplace histogram, we have $\max_i |\tilde{x}_i-x_i|< \gamma$ with probability $1-\beta_0$. Conditioned on this, we make additional error only when $\tilde{x}_i<\theta$ and was truncated to $h_i=0$. When this happens, $x_i < \tilde{x}_i+\gamma < \theta+\gamma$, which holds simultaneously for all $x_i$.
For any $\theta \leq\gamma=\frac{1}{\varepsilon}\ln\frac{|\X|}{\beta_0}$, the error is $O(\frac{1}{\varepsilon}\log\frac{|\X|}{\beta_0})$ with probability $1-\beta_0$.

\textbf{Space and Time.}
We prove for $n\leq |\X|/2$, otherwise it is trivial.
Since there are at most $n$ non-zero $x_i$'s, line~\ref{line:freq} to \ref{line:freqend} take $O(n)$ time and space. For zero-entries $x_1,\dots,x_M$, let $\theta = \frac{1}{\varepsilon}\ln\frac{|\X|}{n}$, we have $p = \frac{n}{2|\X|}$ and $\mathbf{E}[\op{Bin}(M,p)] = Mp \leq \frac{n}{2}$.
Note that $|\X|/2 \leq M\leq |\X|$,
apply Chernoff bound, we have
\[
\Pr[\op{Bin}(M,p)\geq n] \leq e^{-\frac{Mp}{3}}\leq e^{-\frac{n}{12}}\,.
\]
With high probability (in $n$) only $m< n$ values are sampled.
Note that this high-probability result holds for any $\theta \geq \frac{1}{\varepsilon}\ln\frac{|\X|}{n}$, since increasing $\theta$ will reduce $p$, thus reducing the probability above.
Therefore the algorithm uses $O(n)$ space with high probability.

We are left with proving the running time.
Line~\ref{line:bin} samples a binomial random variable $m\sim\op{Bin}(M,p)$. This can be implemented using the BINV algorithm~\cite{DBLP:journals/cacm/KachitvichyanukulS88}, which recursively builds the inverse transformation function, and has running time $O(m)$ for output $m$. Conditioned on $m<n$, this can be done in $O(n)$ time with high probability.
In practice, this is combined with the BTPE algorithm~\cite{DBLP:journals/cacm/KachitvichyanukulS88} to achieve $O(1)$ expected time complexity~\cite{DBLP:journals/corr/abs-2403-11018}.
Line~\ref{line:srs} samples $m$ elements without replacement from $\X-F$, which can be done in $\tilde{O}(n+m)$ time~\cite{DBLP:journals/jpc/BalcerV19}.
To sample each $\hat{x}_j$ from a clipped Laplace distribution, we can apply inverse transform sampling~\cite{gentle2003random}: For $Z'\sim \op{Uniform}[0,1)$, $\op{CDF}_Z^{-1}(Z')$ follows $Z$ since 
\[
\Pr[\op{CDF}_Z^{-1}(Z')\leq z]=\Pr[Z'\leq \op{CDF}_Z(z)] = \op{CDF}_Z(z) = \Pr[Z\leq z]\,.
\]
The CDF of $\hat{x}_j$ is given by $\Pr[\hat{x}_j\leq z]=1-\frac{1}{2p} e^{-\varepsilon z}$, $z\geq\theta$, thus we can generate $\hat{x}_j$ by $\hat{x}_j=\frac{1}{\varepsilon}\ln\frac{1}{2p(1-Z')}$, when $Z'$ distributes by $\op{Uniform}[0,1)$. This can be done in $O(1)$ time.

The value of $\theta$ provides a space-accuracy trade-off: smaller $\theta$ favors accuracy while larger $\theta$ favors space and time. 
\end{proof}
\begin{corollary}\label{thm:linearhist}
There is an $\varepsilon$-DP mechanism for histogram that with high probability runs in $O(n)$ time using $O(n)$ space, and has simultaneous error $O(\frac{1}{\varepsilon}\log |\X|)$. 
\end{corollary}

To compare the pure-DP histogram with the approximate-DP stability-based histogram~\cite{DBLP:journals/jpc/BalcerV19}, both algorithms run in linear time. The $\ell_\infty$ error guarantee is $O(\frac{1}{\varepsilon}\log |\X|)$ for our algorithm, and $O(\frac{1}{\varepsilon} \log\frac{1}{\delta})$ for stability-based histogram~\cite{DBLP:journals/jpc/BalcerV19}, both with constant probability. The pure-DP histogram offers a qualitative better privacy guarantee, and is at least as accurate as
the approximate-DP version when $|\X|=(1/\delta)^{O(1)}$.

Finally, note that both \cite{DBLP:conf/icdt/CormodePST12} and our analysis adopt the Real-RAM model, which assumes constant-time arithmetic operations on infinite-precision floats. In contrast, \cite{DBLP:journals/jpc/BalcerV19} adopts the Word-RAM model on finite computers. Under the latter model, sampling $\op{Bin}(M,p)$ is inefficient since the smallest probability $p^M$ requires $\Omega(M)=\Omega(|X|)$ bits to represent. The best known algorithm takes $\tilde{O}(n^2)$ time to build a pure-DP histogram, and it is open whether a linear time pure-DP histogram algorithm exists under the Word-RAM model.

\begin{corollary}
For constant $d$ and $\eta$, there is an $O(n)$ time algorithm that achieves Theorem~\ref{thm:main}.
\end{corollary}

\section{Optimizations}\label{sec:discuss}

In previous sections, we established matching asymptotic bounds for DBSCAN under DP.
Next, we introduce practical optimizations that are important for the empirical performance.

\subsection{Tighter Upper Bound}\label{sec:lap}

Recall that our algorithm builds upon the noisy upper bound $\widetilde{\op{UB}}(X)$, obtained from taking the sum of $\kappa = \left(1+\frac{8\sqrt{d}}{\eta}\right)^d$ noisy counts $\sum_{X_i\in \mathbf{NB}(X)} \tilde{x}_i$ and a data-independent scale factor $\Gamma=\kappa\gamma$. In fact, we only need to use a high probability upper bound $\widetilde{\op{UB}}$ of $\op{UB}$ for the results to hold. This means the $\Gamma$ function can be tightened.

For a naive Laplace histogram, we apply Theorem~\ref{thm:lapConcentration} with a union bound over all the cells $X\in\X$ to obtain 
\[
\Gamma_{\op{Lap}}=\frac{2\sqrt{2}}{\varepsilon}\max\left\{\sqrt{\kappa\ln\frac{2|\X|}{\beta}}, \ln\frac{2|\X|}{\beta}\right\}\,,
\]
which improves the bound by a constant factor of $\sqrt{\kappa}$. For the linear time histogram in Section~\ref{sec:hist}, recall the error of each entry is bounded by $\theta+\gamma$. With Laplace concentration, the error bound for the independent sum is $\Gamma=\kappa \theta + \Gamma_{\op{Lap}}$, which gives the exact guarantee of our mechanism.

\begin{theorem}[Exact Bound]
For $0<\eta \leq 4$, there is a linear time  $\varepsilon$-DP mechanism for $(\alpha,\MinPts)$-DBSCAN clustering that is $(3+\eta,\tau;\beta)$-accurate, where 
\[
\tau = \frac{2\kappa}{\varepsilon}\cdot \max\left\{0, \ln {\frac{|X|}{n}}\right\} + \frac{4\sqrt{2}}{\varepsilon}\cdot \max\left\{\sqrt{\kappa \ln\frac{2|\X|}{\beta}}, \ln\frac{2|\X|}{\beta}\right\}\,,
\]
where $\kappa = \left(1+\frac{8\sqrt{d}}{\eta}\right)^d$ and $|\X|=\left(\frac{4\sqrt{d}}{ \eta \alpha}\right)^d$.
\end{theorem}

Note that when dealing with frequencies in integers, an alternative is to use Geometric noise in Theorem~\ref{thm:geom}, whose error bound is given in Theorem~\ref{thm:geomConcentration}.
For our problem where $\kappa$ is small, the Laplace bound is more tight when $\varepsilon$ is small. For example, when $\varepsilon=1$, $\beta=1/3$, $\kappa=21$ and $|\X|=1000$, Theorem~\ref{thm:geomConcentration} is forced to use the first bound to get  $\Gamma_{\op{Geom}}=252$, whereas the Laplace bound is only
$\Gamma_{\op{Lap}}=38.2$.
This similarly applies to  approximate $(\varepsilon,\delta)$-DP, where Gaussian noises could replace Laplace noise. With Gaussian concentration, it error bound is 
\[
\Gamma_{\op{Gauss}} = \frac{2}{\varepsilon}\sqrt{\kappa\ln\frac{1.25}{\delta}\ln\frac{2|\X|}{\beta}}\,.
\]
For reasonable $\delta\ll 0.1$, we have $\ln(1.25/\delta)>2$, and the bound is worse than the pure-DP version using Laplace mechanism. Note that the error function $\Gamma$ of a DP histogram directly translates to the additive error $\tau$ in the approximation ratio of DP-DBSCAN.

\subsection{Choice of Parameters}\label{sec:parameters}

Choosing the parameters (namely $\alpha$ and $\MinPts$) is an essential task even for the original DBSCAN algorithm.
It is known that while DBSCAN is sensitive to the radius $\alpha$, the $\MinPts$ parameter does not have a large effect on the clustering quality~\cite{DBLP:journals/tods/SchubertSEKX17}. 
For our DP-DBSCAN mechanism, we will be using the original $\alpha$ with a scaled up $\MinPts'=\MinPts+\tau$, compared to the parameters $(\alpha,\MinPts)$ in the non-private setting. 
That is, our algorithm will find all $(\alpha,\MinPts+\tau)$-clusters, while guaranteeing the resulting spans are from clusters of size $(3+\eta, \MinPts)$-approximation. This is both to guarantee that $\MinPts'>\tau$, and to offset the effect that $\op{UB}(X)$ itself without noise is an upper bound that tends to over-estimate the density of points in $X$ already.

\section{Experiments}
\label{sec:exp}

In this section, we present both visual and quantitative results to demonstrate the effectiveness of our algorithm.

\subsection{Experiment Setup}

\begin{table}[htbp]
    \centering
    \caption{Summary of Datasets Evaluated ($\alpha$ wrt.~data scale).}
    \label{tab:dataset}
    \begin{tabular}{ccccc}
    \toprule
        & Dataset & $n$  & Clusters & $(\alpha,\MinPts)$ \\
    \midrule
        \multirow{3}{*}{Synthetic} & Circles~\cite{scikit-learn} &  2,000 & 2 & $(0.2,10)$ \\
        & Moons~\cite{scikit-learn} &  2,000 & 2 & $(0.2,7)$ \\
        & Blobs~\cite{scikit-learn} &  2,000 & 3 & $(0.2,7)$ \\
    \midrule
        \multirow{3}{*}{Benchmark} & Cluto-t4~\cite{karypis2002cluto} &  8,000 & 6$^\dagger$ & $(9.0,11)$ \\
        & Cluto-t5~\cite{karypis2002cluto} & 8,000 & 6$^\dagger$ & $(9.0,20)$ \\
        & Cluto-t7~\cite{DBLP:journals/computer/KarypisHK99} &  10,000 & 9$^\dagger$ & $(12.0, 20)$ \\
    \midrule
        \multirow{4}{*}{Real} & Crash~\cite{DataCrash} &  1,860,785 & - & $(100\text{m},300)$ \\
        & Cabs-tiny~\cite{c7j010-22} &  845,685 & - & $(50\text{m},1000)$ \\
        & Cabs~\cite{c7j010-22} &  10,995,626 & - & $(20\text{m},500)$ \\
        & HAR70+~\cite{DBLP:journals/sensors/UstadLTTVBS23} & 103,860 & 7 & (0.01, 5)\\
    \bottomrule
     \multicolumn{5}{l}{$^{\dagger}$\footnotesize{contains noise}}
    \end{tabular}
\end{table}
\textbf{Datasets.} We conducted experiments on both synthetic and real-world datasets.
Table~\ref{tab:dataset} summaries the details of the datasets we evaluated, including the number of clusters for synthetic datasets, and the DBSCAN parameters with respect to the original data.

The \textbf{Circles}, \textbf{Moons}, and \textbf{Blobs} datasets are synthetic datasets generated by scikit-learn~\cite{scikit-learn}, commonly used to  evaluate clustering algorithms.
The \textbf{Cluto} (a.k.a.~Chameleon) benchmark~\cite{karypis2002cluto,DBLP:journals/computer/KarypisHK99} contains non-convex clusters and noises, making it particularly challenging for centroid-based clustering algorithms.

For real-world data, we use two geographic datasets to demonstrate that our DP-DBSCAN algorithm can be applied to visualize private data. The \textbf{Crash} dataset~\cite{DataCrash} contains motor vehicle collision data in New York City from 2012 to 2024. After cleaning the data by filtering valid event locations within New York City (between $40^\circ$ N $73^\circ$ W and $41^\circ$ N $75^\circ$ W), we retained 1.8 million records out of the original 2.1 million. We also used the \textbf{Cabs} dataset~\cite{c7j010-22}, which consists of trajectory data for taxi cabs in San Francisco, collected from around 500 taxis over 30 days.
We filtered cab locations between $37.4^\circ$ N $122.1^\circ$ W and $37.9^\circ$ N $122.7^\circ$ W, resulting in 11 million records.
Since the dataset is large enough to cause the original DBSCAN out-of-memory, we also use a subset of the data \textbf{Cabs-tiny}, containing the starting and ending points of each passenger-occupied trajectory in a smaller region. This gives 800k records.

For both geographic real-world datasets, we projected the longitude and latitude coordinates to kilometers when calculating Euclidean distances. Each degree of latitude corresponds to approximately 111.2 km, and each degree of longitude is approximately 85.2 km at around 40°N. We then ran DBSCAN with the distance parameter $\alpha$ in meters. 

Finally, to evaluate the performance of DP-DBSCSAN for classification tasks, we use a Human Activity Recognition (\textbf{HAR}) dataset~\cite{DBLP:journals/sensors/UstadLTTVBS23} obtained from UCI Machine Learning Repository~\cite{Dua:2019}. The dataset contains accelerometer recordings for 18 aged participants, where each record consists of 3D acceleration of both the back and the thigh sensor. The target is to recognize the corresponding activity, e.g.~walking, standing, sitting. We use the 3D acceleration of back sensor in our experiments. For each person, there are roughly 100k records.

\textbf{Metrics.} For datasets with true labels, we evaluate the clustering accuracy using two commonly adopted measures: the Adjusted Rand Index (ARI)~\cite{hubert1985comparing} and Adjusted Mutual Information
(AMI)~\cite{DBLP:conf/icml/NguyenEB09}.
Both metrics range from $[-0.5,1.0]$, where a score close to 0 indicates random labeling, and a score of 1 signifies perfect agreement with the ground truth.

\textbf{Mechanisms.} We implemented our proposed DP-DBSCAN mechanism in Python and compared it with: 1) the original DBSCAN and k-means mechanism implemented in scikit-learn~\cite{scikit-learn}, and 2) the DP-Kmeans mechanism from IBM’s Diffprivlib~\cite{diffprivlib}.
We also implemented existing methods~\cite{DBLP:conf/iiki/NiLLBY17,wu2015dp,DBLP:journals/access/NiLWJY18,jin2019improved} that privatize the pairwise distance matrix. Since they have $\ARI$ and $\AMI=0$ on all the datasets, we omit them in the comparison below.

Unless otherwise specified, we run Kmeans-based algorithms with the exact number of clusters if it is known, or $k=3$ if not; and run DBSCAN-based algorithms with the parameters specified in Table~\ref{tab:dataset}.
For DP-DBSCAN, we make our algorithm adopt for the linear-time histogram in Section~\ref{sec:hist} when $n\leq |\X|/2$, and the $O(|\X|)$ time naive Laplace histogram if the opposite is true.
We compare the difference in their running time in Section~\ref{sec:time}.

The code and datasets for all the experiments can be found at \url{https://github.com/QiuTedYuan/DpDBSCAN}.
To ensure reproducibility, we initialized the random seed to 0 for all experiments in the figures, and the reported scores are the average of three trials with random seeds 0, 1, and 2.
All experiments were conducted on a desktop equipped with an Intel(R) Core(TM) i7-13700KF 3.4 GHz CPU and 64GB of memory.

\subsection{Visualization of DP-DBSCAN}

\begin{figure*}[htbp]
    \centering
     \begin{subfigure}[t]{0.19\textwidth}
         \centering
         \includegraphics[width=\textwidth]{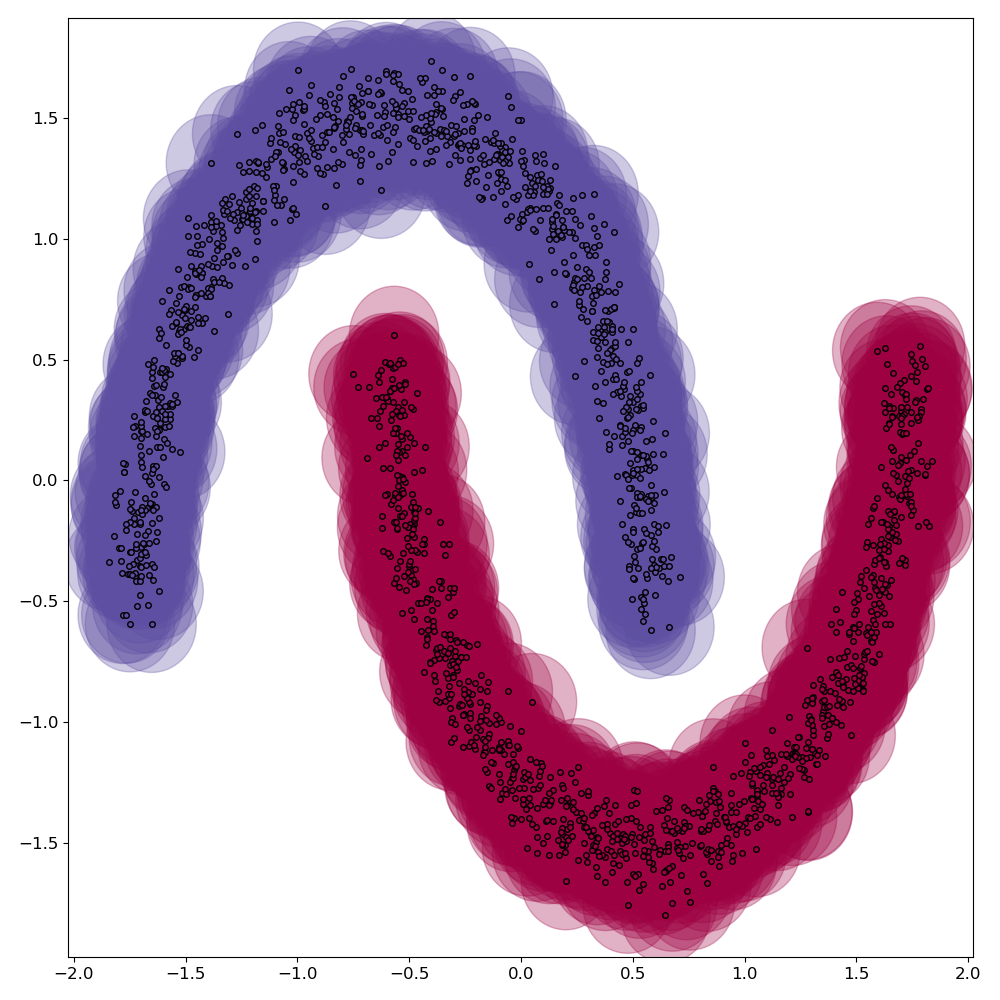}
         \caption{DBSCAN}
         \label{fig:moons_dbscan}
     \end{subfigure} \hfill
     \begin{subfigure}[t]{0.19\textwidth}
         \centering
         \includegraphics[width=\textwidth]{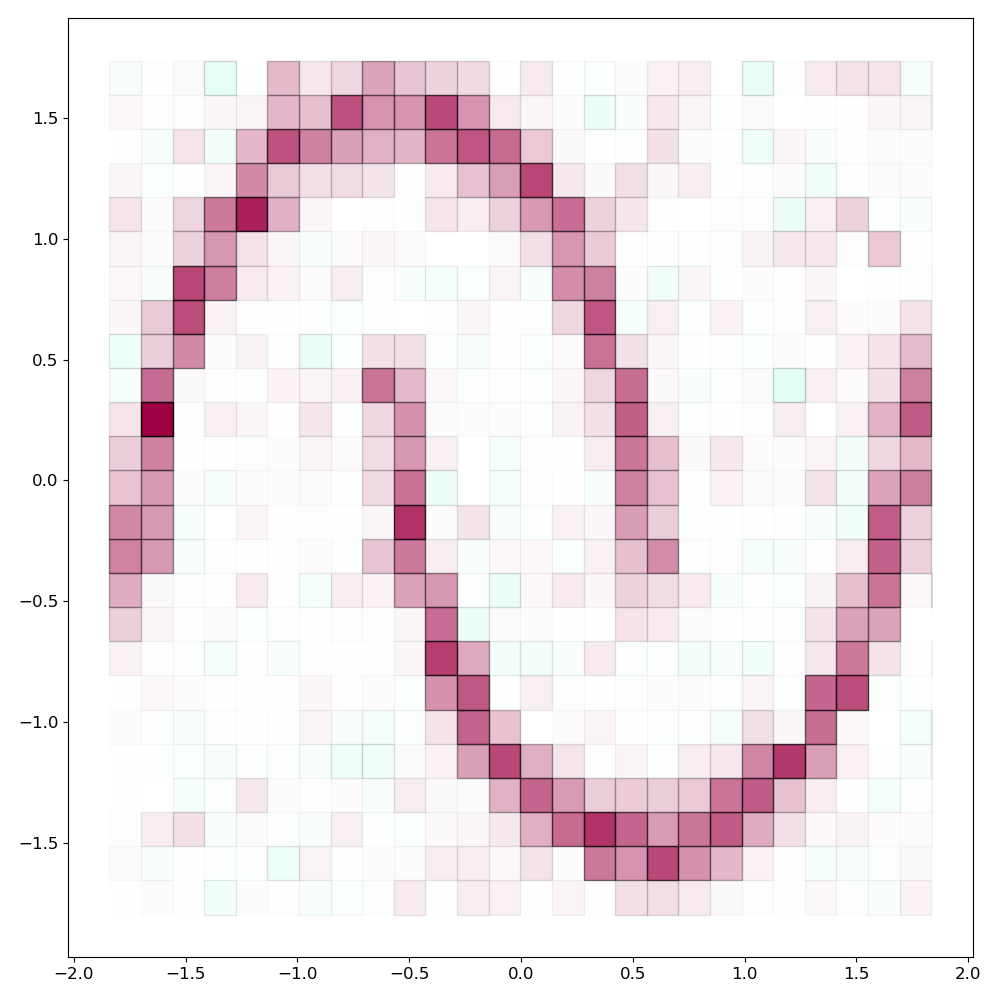}
         \caption{Noisy histogram}
         \label{fig:moons_counts}
     \end{subfigure} \hfill
     \begin{subfigure}[t]{0.19\textwidth}
         \centering
         \includegraphics[width=\textwidth]{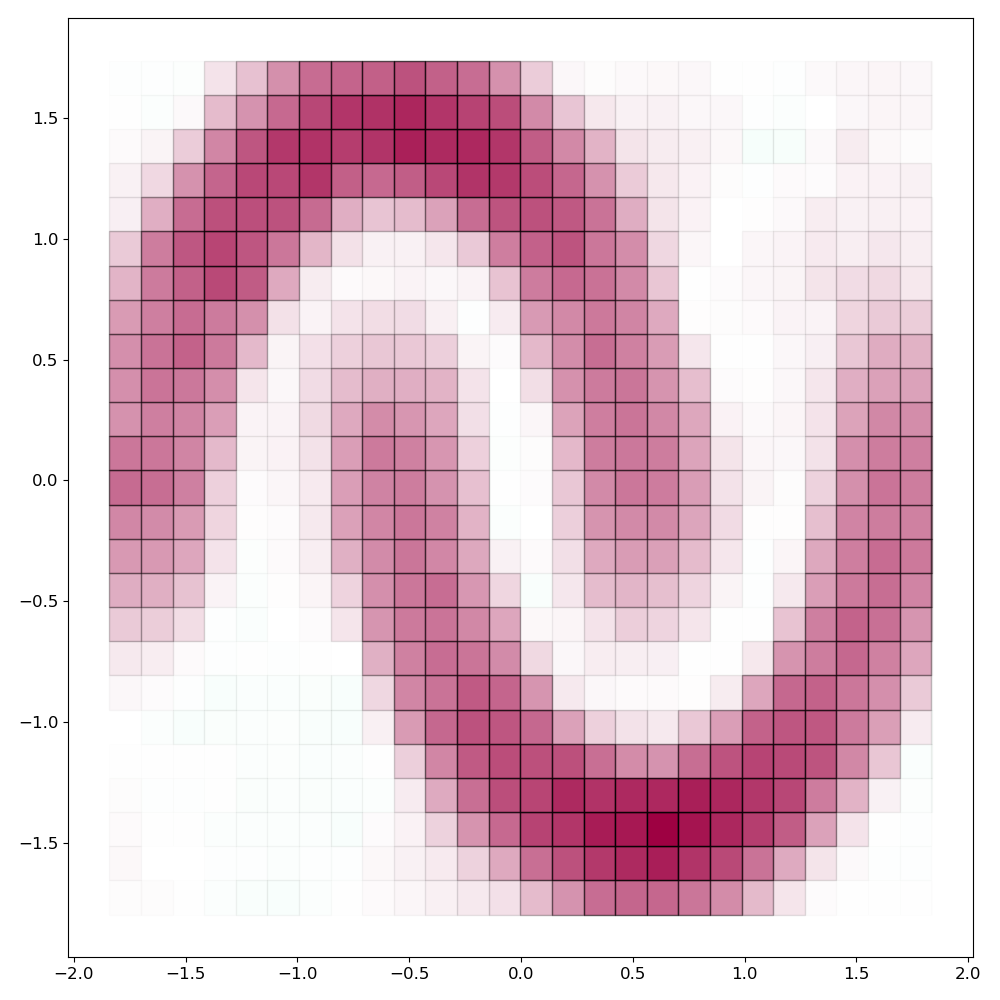}
         \caption{Neighbor sums}
         \label{fig:moons_sum}
     \end{subfigure} \hfill
     \begin{subfigure}[t]{0.19\textwidth}
         \centering
         \includegraphics[width=\textwidth]{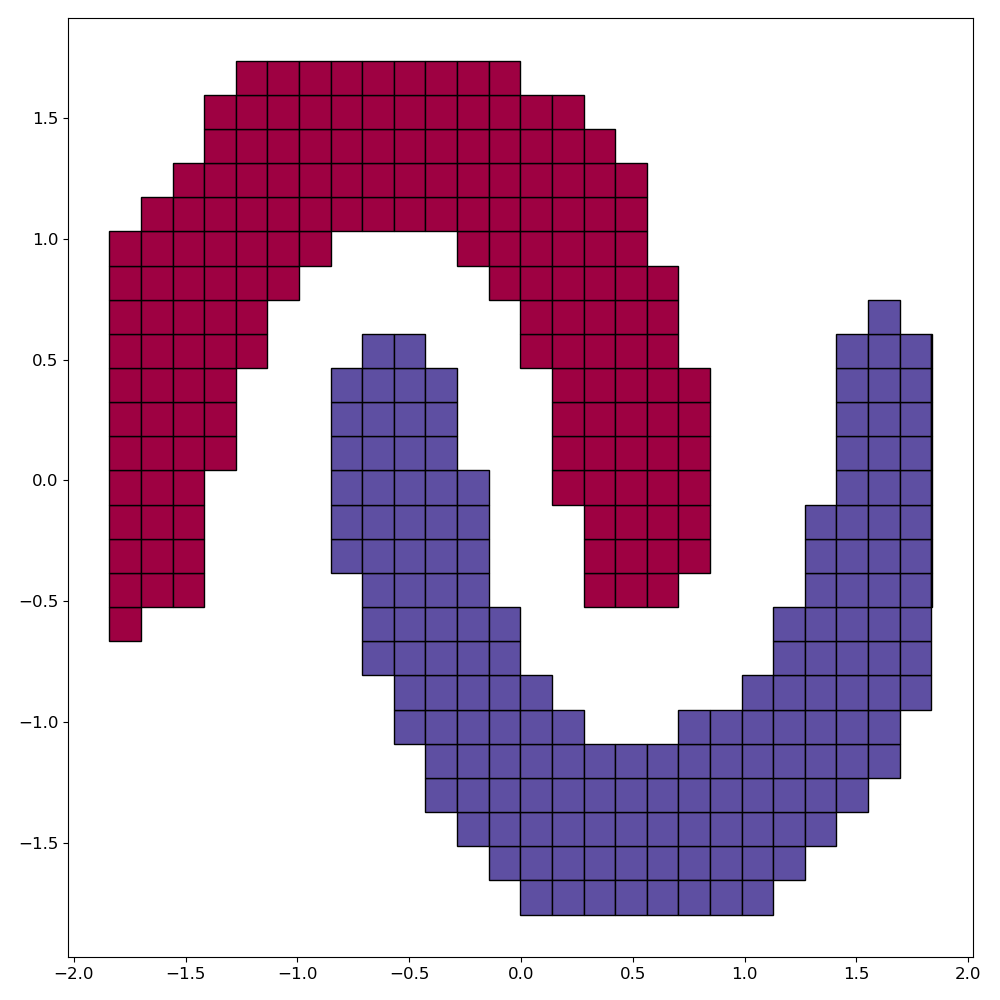}
         \caption{Approx.~spans}
         \label{fig:moons_dp_dbscan_grids}
     \end{subfigure} \hfill
     \begin{subfigure}[t]{0.19\textwidth}
         \centering
         \includegraphics[width=\textwidth]{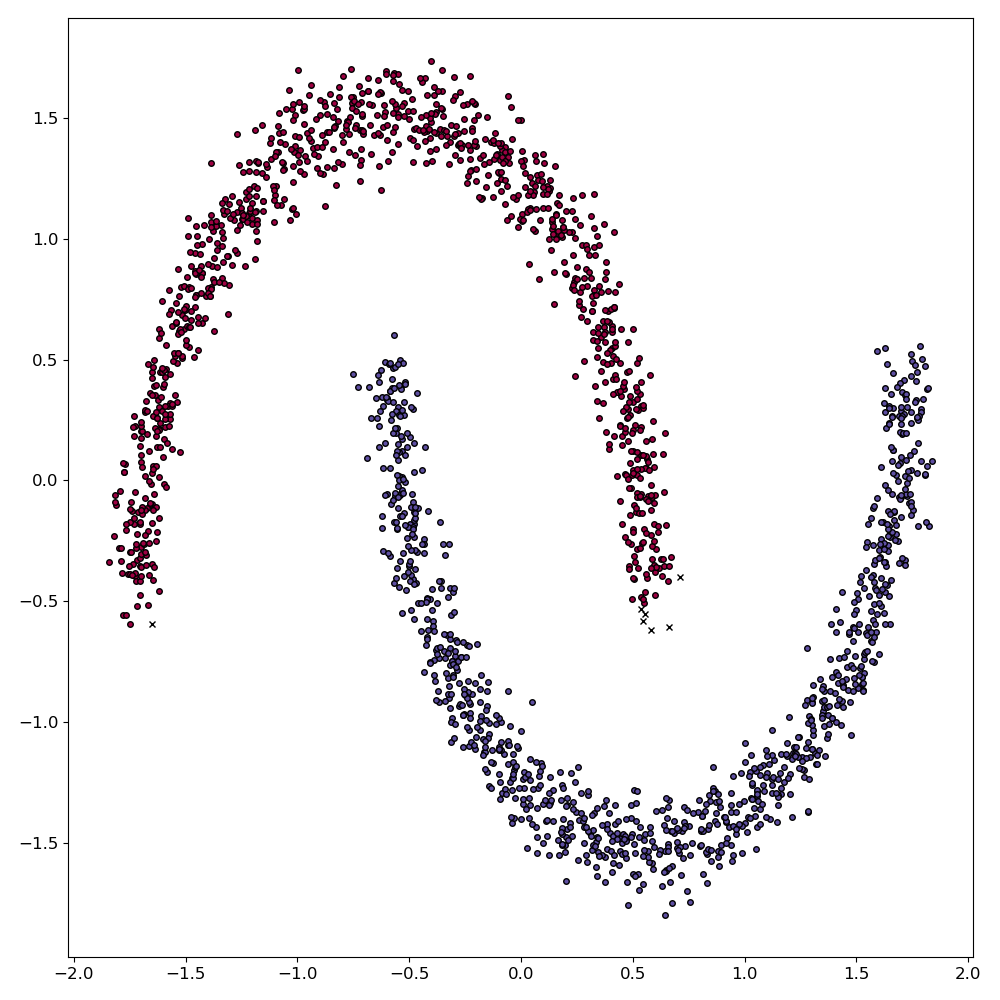}
         \caption{Extracted labels}
         \label{fig:moons_dp_dbscan_pts}
     \end{subfigure}

    \caption{Visualization of DP-DBSCAN on Moons Dataset.}
    \label{fig:moons}
\end{figure*}

We first illustrate how our mechanism works on the Moons dataset in Figure~\ref{fig:moons}. Given the dataset, we construct cells with width $w=\frac{\alpha}{\sqrt{d}}=0.14$, and release the noisy histogram under $\varepsilon=1$ in Figure~\ref{fig:moons_counts}, where cell transparency represents frequency. Cells with positive/negative counts are shown in red/green.
Next, we compute the sum of noisy counts for each cell and its neighboring cells, as seen in Figure~\ref{fig:moons_sum}. As indicated in Figure~\ref{fig:cell}, each noisy sum consists of the counts from $\kappa=21$ neighboring grids in this case, including the grid itself. 
Finally, by identifying core grids and merging adjacent ones, Algorithm~\ref{algo} produces the differentially private approximation to the DBSCAN cluster spans, as depicted in Figure~\ref{fig:moons_dp_dbscan_grids}.

To analyze the quality of the clusters, we calculate the labels assigned to each point in Figure~\ref{fig:moons_dp_dbscan_pts}, which are not a private output of the mechanism but for evaluation only.
It can be observed that while most points receive correct labels, points at the boarders are classified as noises in DP-DBSCAN, whereas they are core points in the original DBSCAN. This is reasonable under differential privacy: the boarder points are highly sensitive with respect to the addition or removal of a single point, whereas interior points are still highly likely to remain core, even if some of its neighbors gets removed in a neighboring dataset. Quantitatively, Our mechanism achieves $\AMI=0.97$ and $\ARI=0.94$.

\begin{figure*}[htbp]
    \centering
    \begin{subfigure}[t]{0.19\textwidth}
         \centering
         \includegraphics[width=\textwidth]{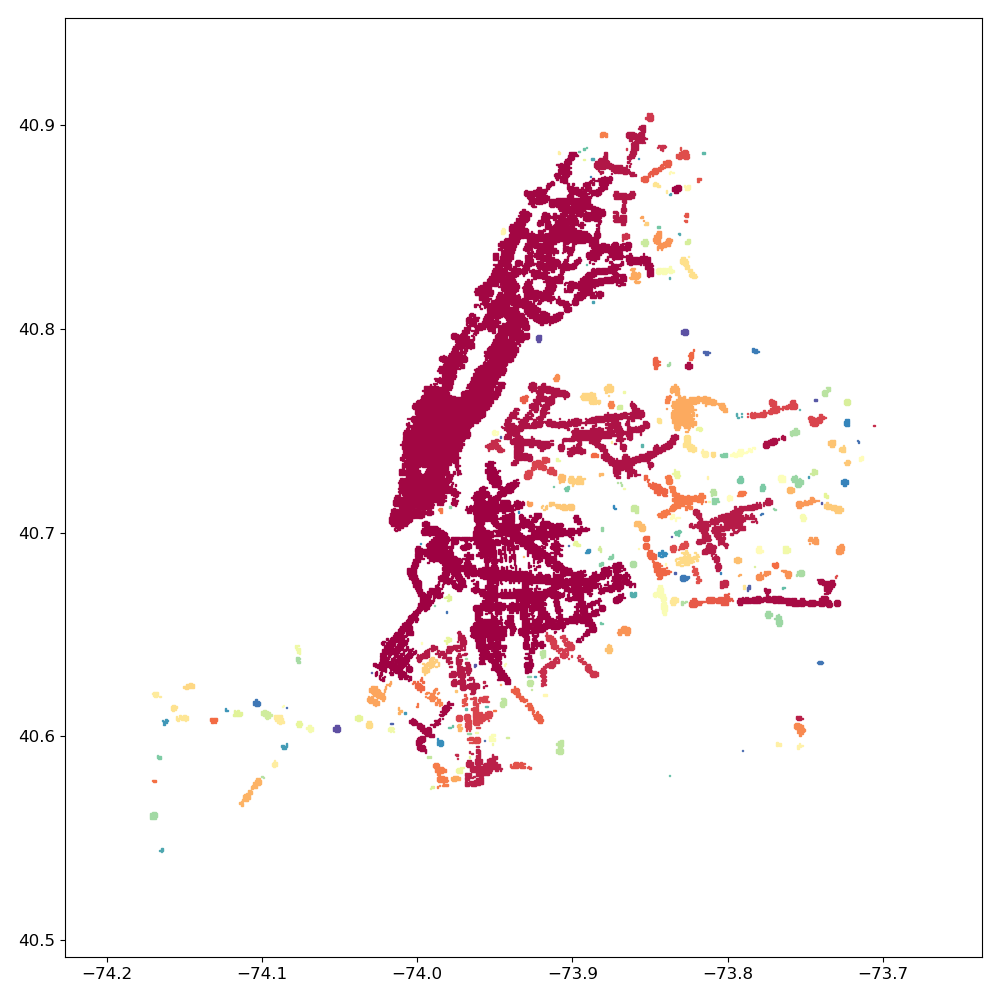}
         \caption{Crash/$O(n)$}
         \label{fig:crash_cluster_linear}
    \end{subfigure} \hfill
    \begin{subfigure}[t]{0.19\textwidth}
         \centering
         \includegraphics[width=\textwidth]{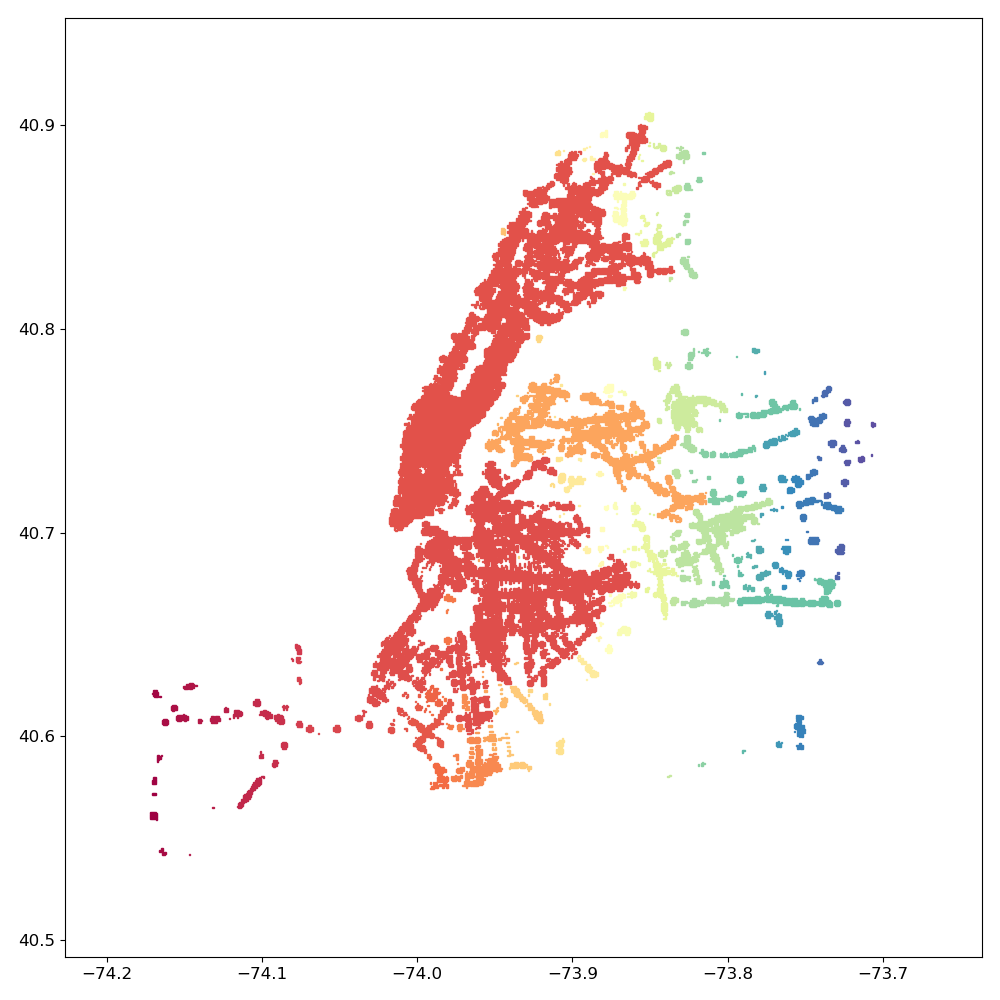}
         \caption{Crash/$O(|\X|)$}
         \label{fig:crash_cluster_naive}
    \end{subfigure} \hfill
    \begin{subfigure}[t]{0.19\textwidth}
         \centering
         \includegraphics[width=\textwidth]{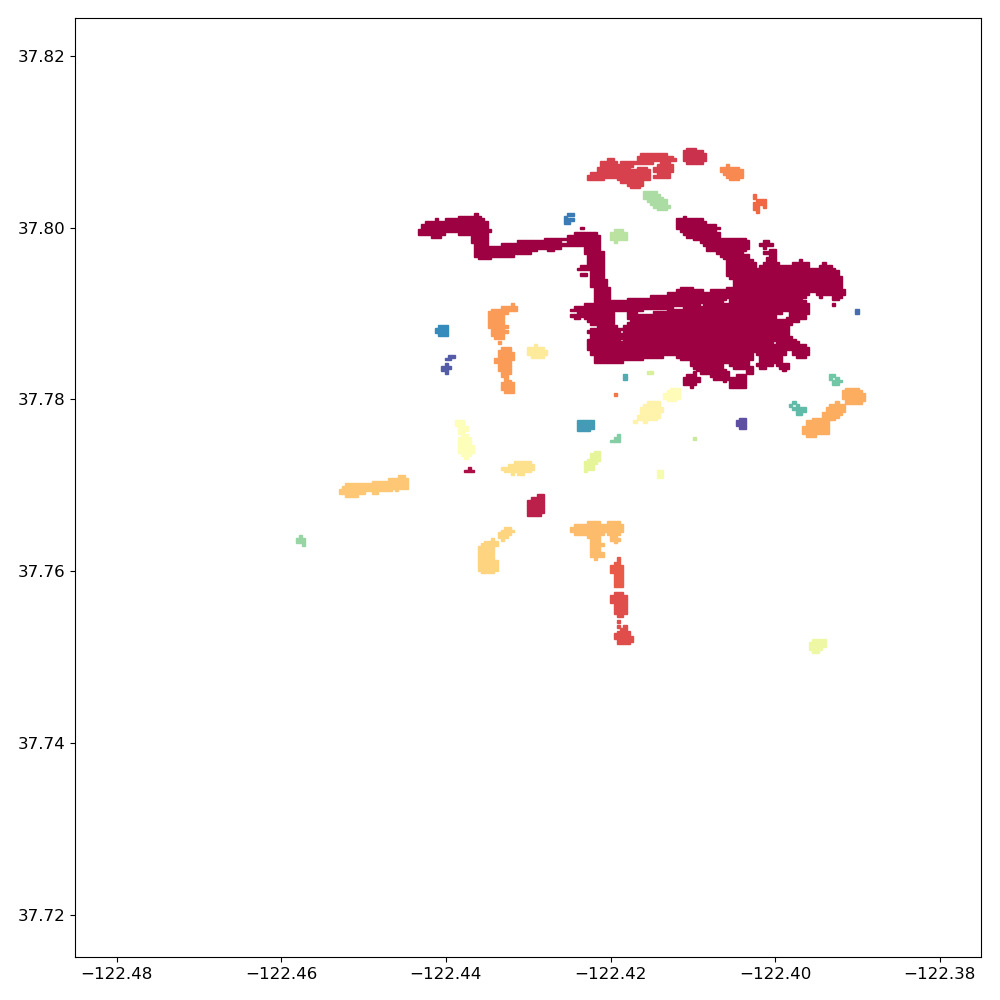}
         \caption{Cabs-tiny/$O(n)$}
    \label{fig:cabs_cluster_linear}
    \end{subfigure} \hfill
    \begin{subfigure}[t]{0.19\textwidth}
         \centering
         \includegraphics[width=\textwidth]{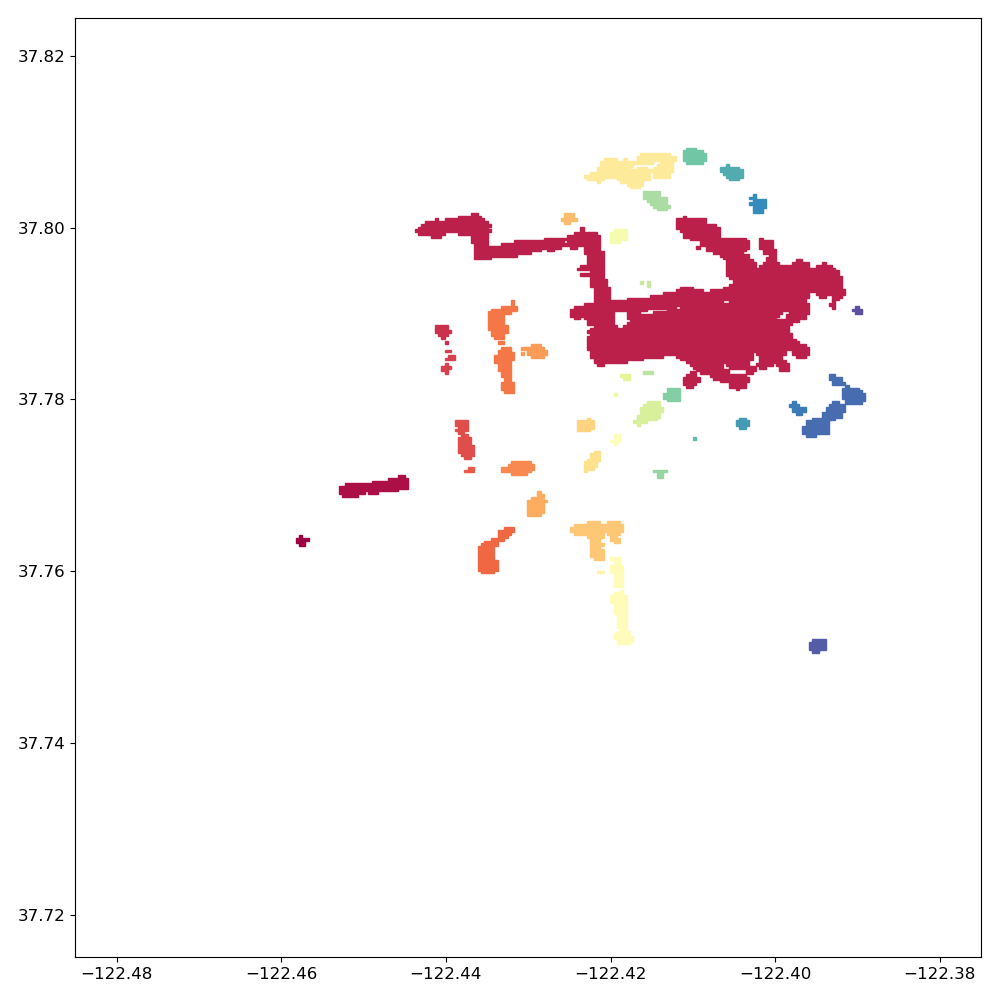}
         \caption{Cabs-tiny/$O(|\X|)$}
    \label{fig:cabs_cluster_naive}
    \end{subfigure} \hfill
    \begin{subfigure}[t]{0.19\textwidth}
         \centering
         \includegraphics[width=\textwidth]{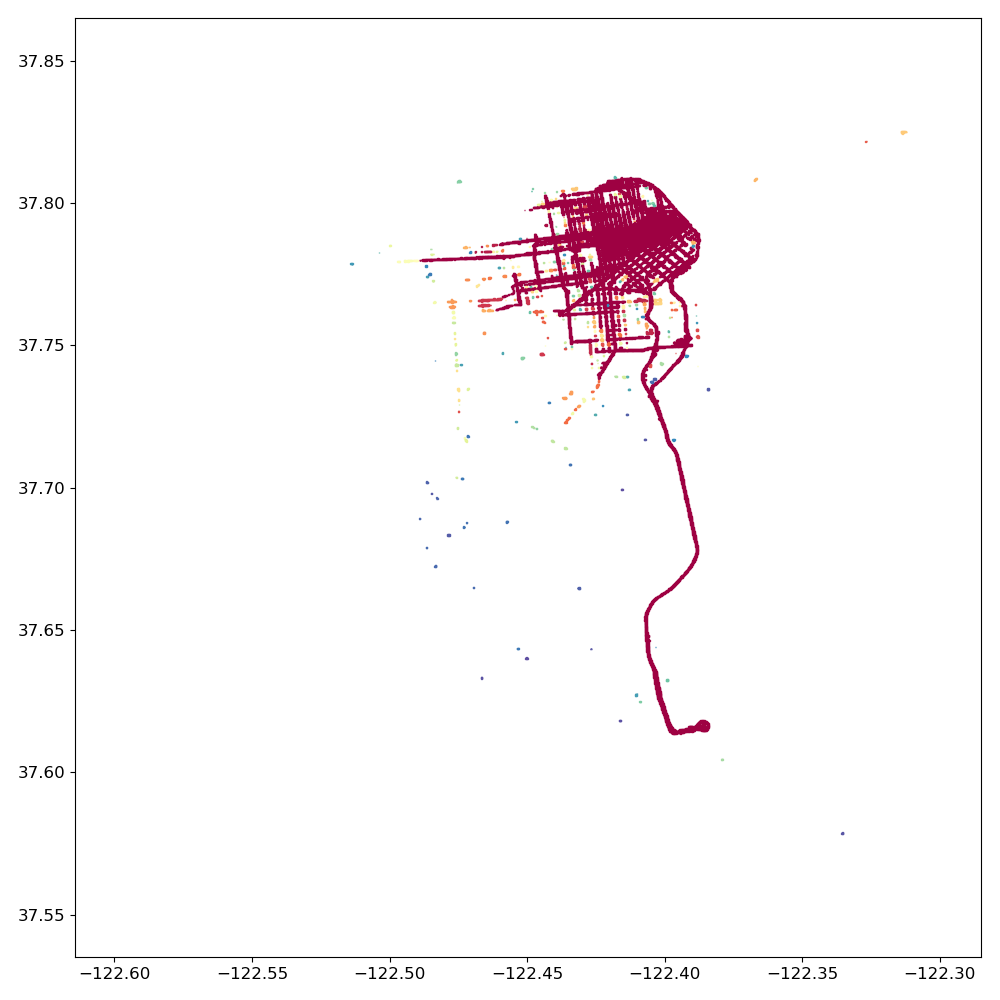}
         \caption{Cabs/$O(n)$}
    \label{fig:cabs_cluster}
    \end{subfigure} 
    \caption{Comparison of DP-DBSCAN clusters on geographic datasets.}
    \label{fig:real}
\end{figure*}

We also visualize the clustering results of DP-DBSCAN on real-world datasets in Figure~\ref{fig:real}, where $\varepsilon=1$.
The point labels are omitted due to the high number of points and lack of ground truth.
For the Crash dataset, 284 clusters are detected by our algorithm in Figure~\ref{fig:crash_cluster_linear}, with the largest cluster spanning Manhattan, Brooklyn, and the Bronx. The performance is similar compared to the $O(|\X|)$ time implementation in Figure~\ref{fig:crash_cluster_naive}, which assigned 271 clusters.
This justifies the accuracy of our linear-time histogram. Their running time will be evaluated in Section~\ref{sec:time}.
For the Cabs-tiny dataset, we observe similar performance in Figure~\ref{fig:cabs_cluster_linear} and \ref{fig:cabs_cluster_naive}, where both implementations of our DP-DBSCAN output 38 clusters, highlighting the high frequency of taxi pick-ups in the northeastern part of San Francisco. However, since we use a small radius $\alpha=20$m on the raw dataset, the naive implementation fails to  work. The linear time implementation can still be run efficiently to produce Figure~\ref{fig:cabs_cluster}, which identified a route to the SFO airport.

\subsection{Comparison of Accuracy}

\begin{figure*}[htbp]
    \centering
    \begin{subfigure}[b]{0.32\textwidth}
         \centering
         \includegraphics[width=\textwidth]{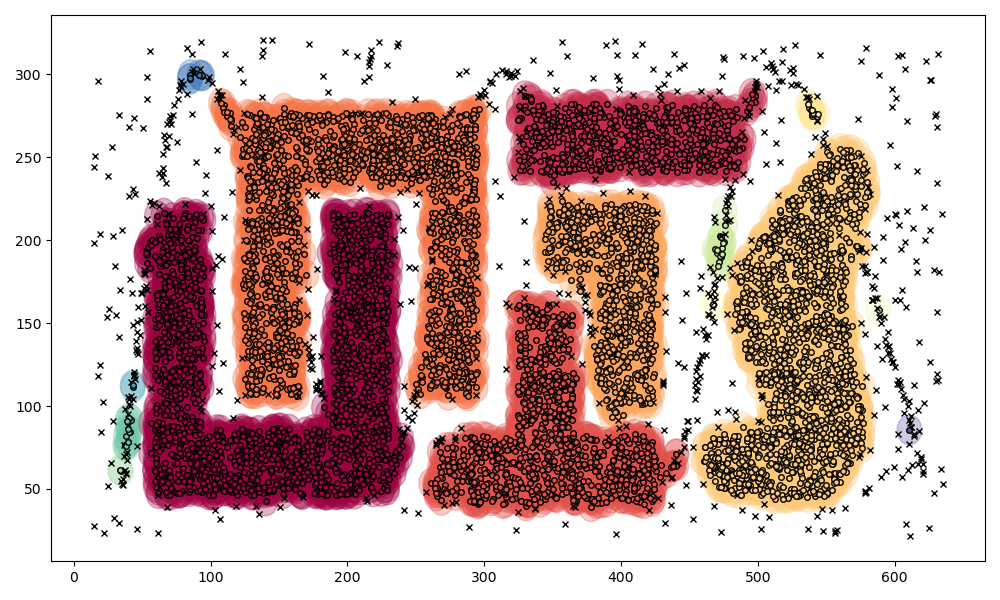}
         \caption{DBSCAN (Cluto-t4)}
         \label{fig:cluto_t4_dbscan}
     \end{subfigure} \hfill
     \begin{subfigure}[b]{0.32\textwidth}
         \centering
         \includegraphics[width=\textwidth]{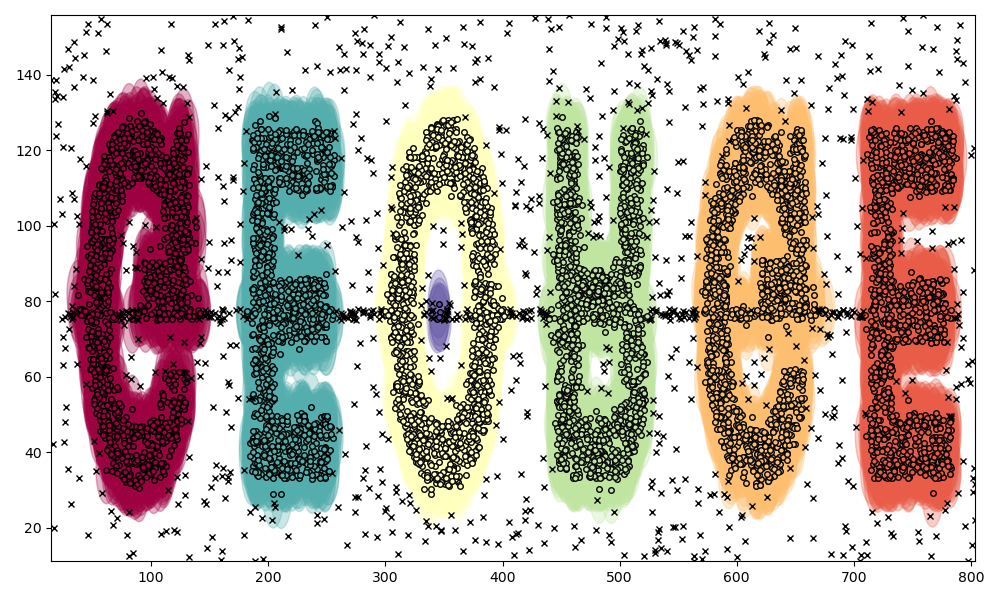}
         \caption{DBSCAN (Cluto-t5)}
         \label{fig:cluto_t5_dbscan}
     \end{subfigure}
     \hfill
     \begin{subfigure}[b]{0.32\textwidth}
         \centering
         \includegraphics[width=\textwidth]{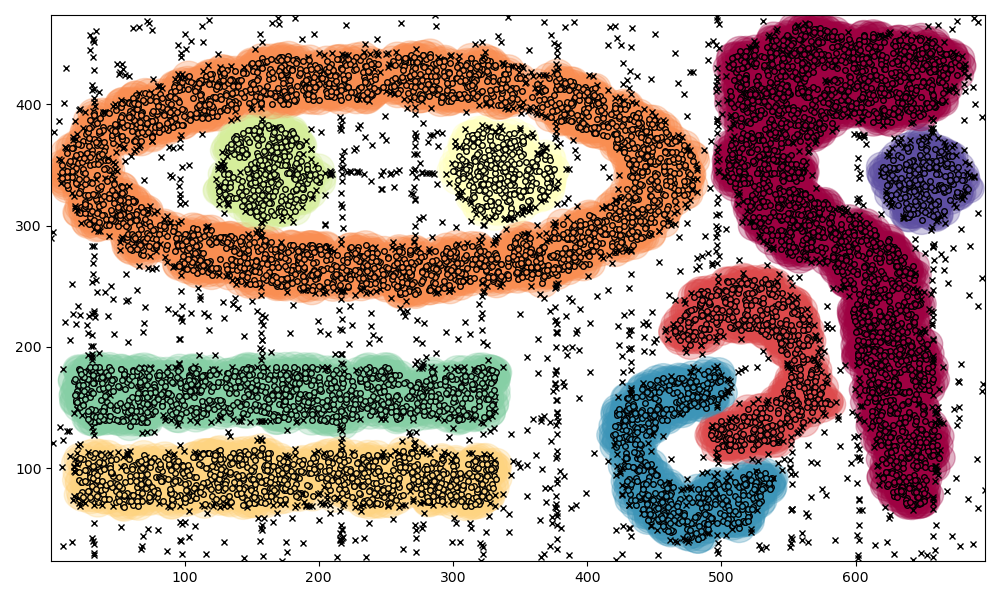}
         \caption{DBSCAN (Cluto-t7)}
         \label{fig:cluto_t7_dbscan}
    \end{subfigure} 
    
    \begin{subfigure}[b]{0.32\textwidth}
         \centering
         \includegraphics[width=\textwidth]{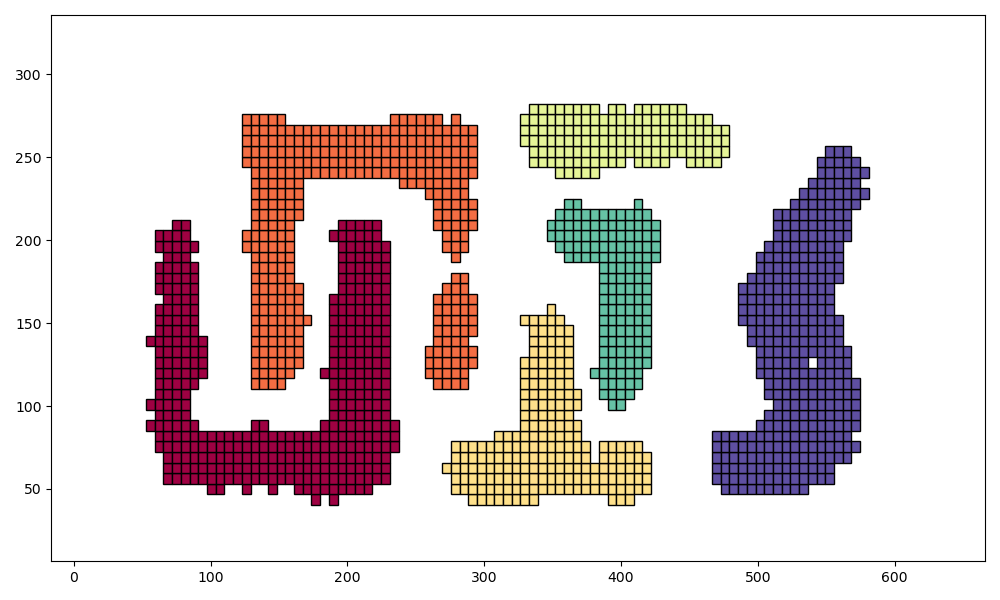}
         \caption{DP-DBSCAN (Cluto-t4)}
         \label{fig:cluto_t4_dp_dbscan}
     \end{subfigure} \hfill
     \begin{subfigure}[b]{0.32\textwidth}
         \centering
         \includegraphics[width=\textwidth]{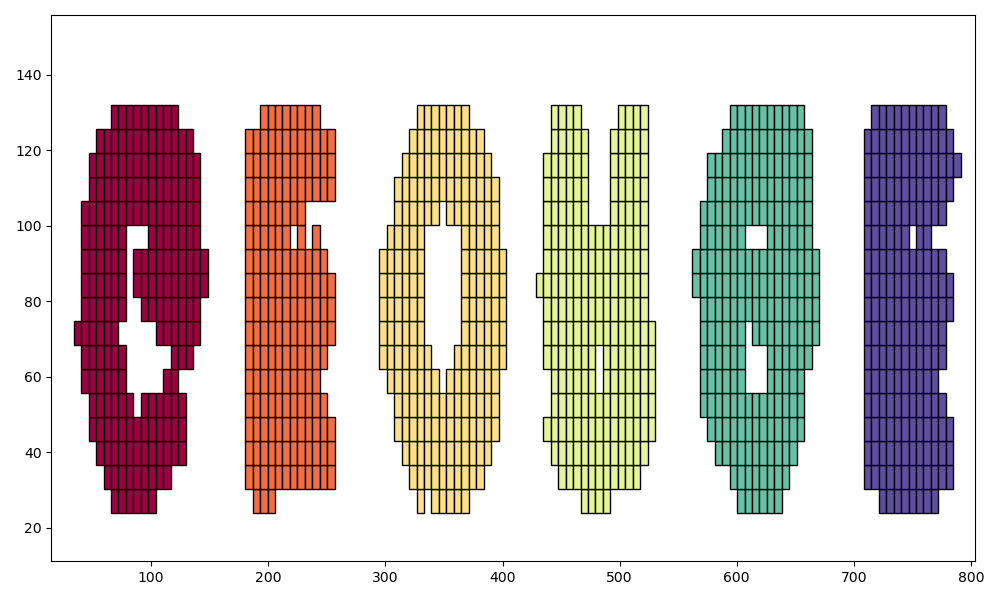}
         \caption{DP-DBSCAN (Cluto-t5)}
         \label{fig:cluto_t5_dp_dbscan}
     \end{subfigure} \hfill
     \begin{subfigure}[b]{0.32\textwidth}
         \centering
         \includegraphics[width=\textwidth]{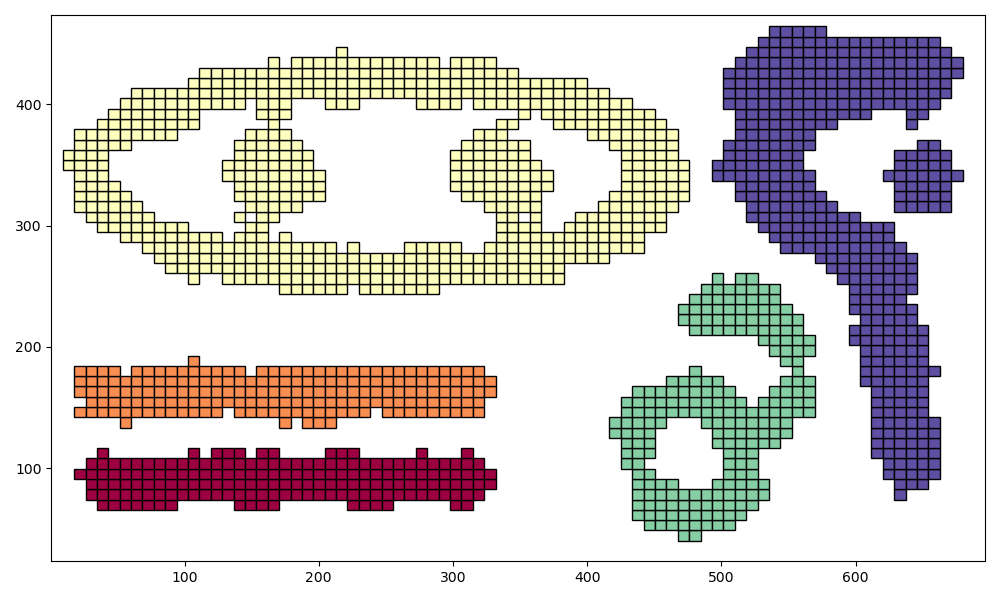}
         \caption{DP-DBSCAN (Cluto-t7)}
         \label{fig:cluto_t7_dp_dbscan}
    \end{subfigure} 
    
    \begin{subfigure}[b]{0.32\textwidth}
         \centering
         \includegraphics[width=\textwidth]{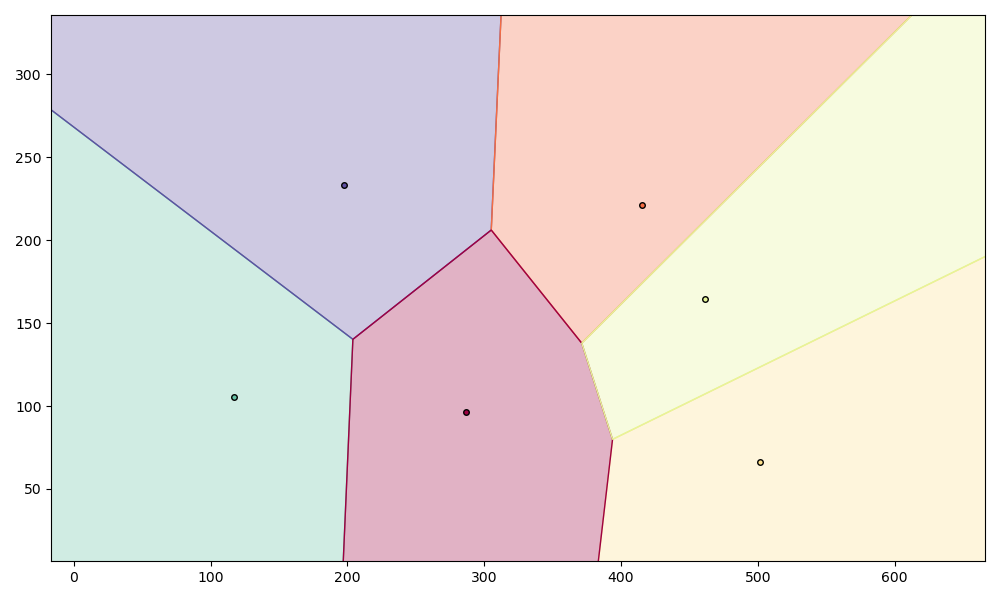}
         \caption{DP-Kmeans (Cluto-t4)}
         \label{fig:cluto_t4_dp_kmeans}
     \end{subfigure} \hfill
     \begin{subfigure}[b]{0.32\textwidth}
         \centering
         \includegraphics[width=\textwidth]{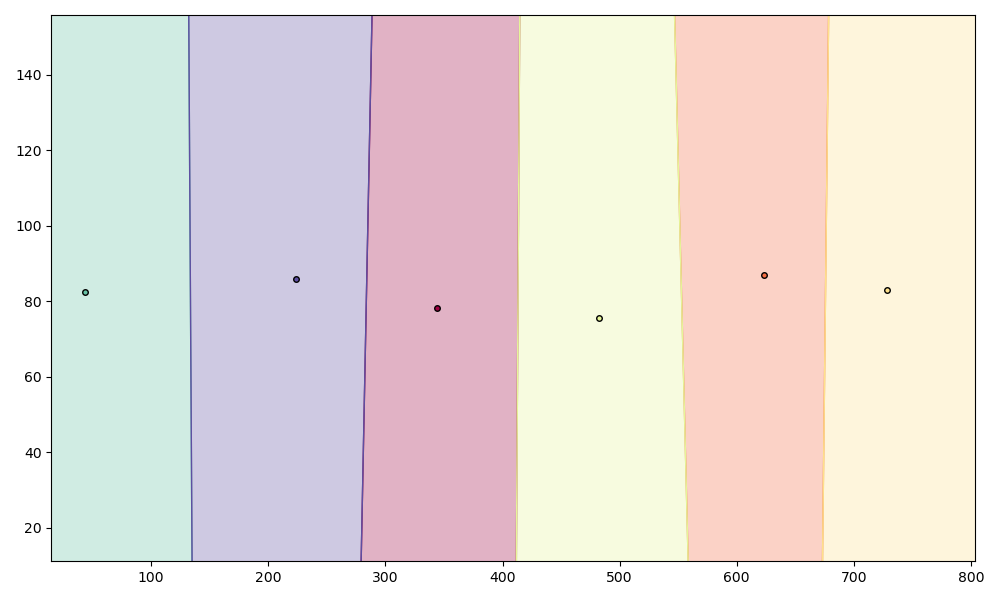}
         \caption{DP-Kmeans (Cluto-t5)}
         \label{fig:cluto_t5_dp_kmeans}
     \end{subfigure} \hfill
     \begin{subfigure}[b]{0.32\textwidth}
         \centering
         \includegraphics[width=\textwidth]{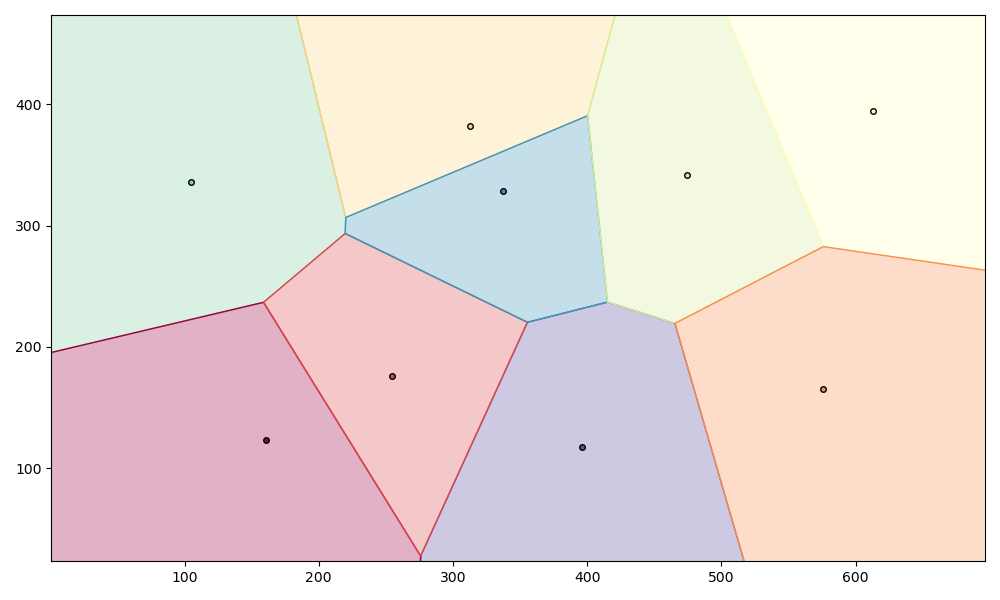}
         \caption{DP-Kmeans (Cluto-t7)}
         \label{fig:cluto_t7_dp_kmeans}
    \end{subfigure}
    
      \begin{subfigure}[b]{0.32\textwidth}
         \centering
         \includegraphics[width=\textwidth]{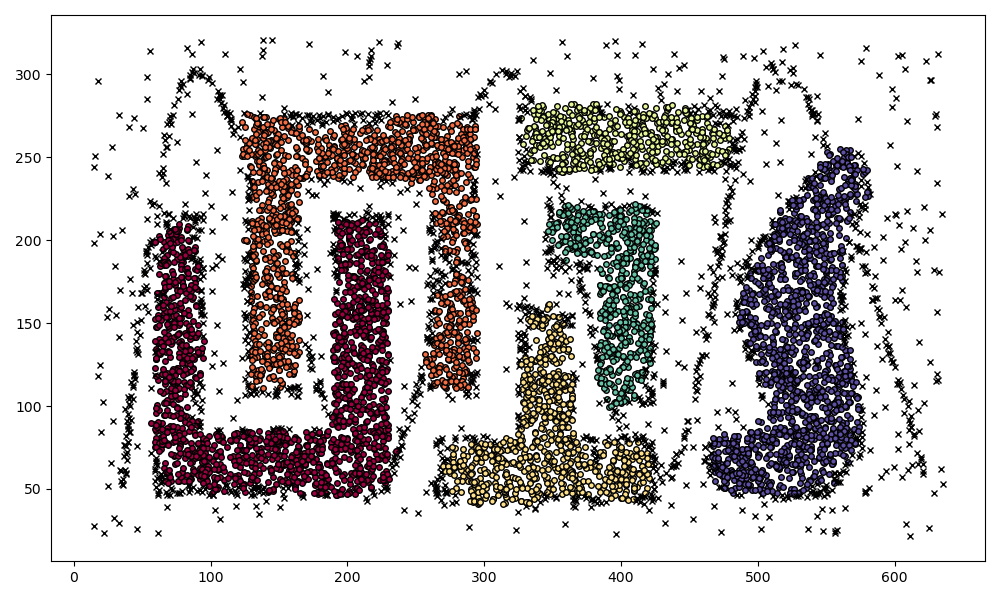}
         \caption{DP-DBSCAN labels (Cluto-t4)}
         \label{fig:cluto_t4_dp_dbscan_labels}
     \end{subfigure} \hfill
     \begin{subfigure}[b]{0.32\textwidth}
         \centering
         \includegraphics[width=\textwidth]{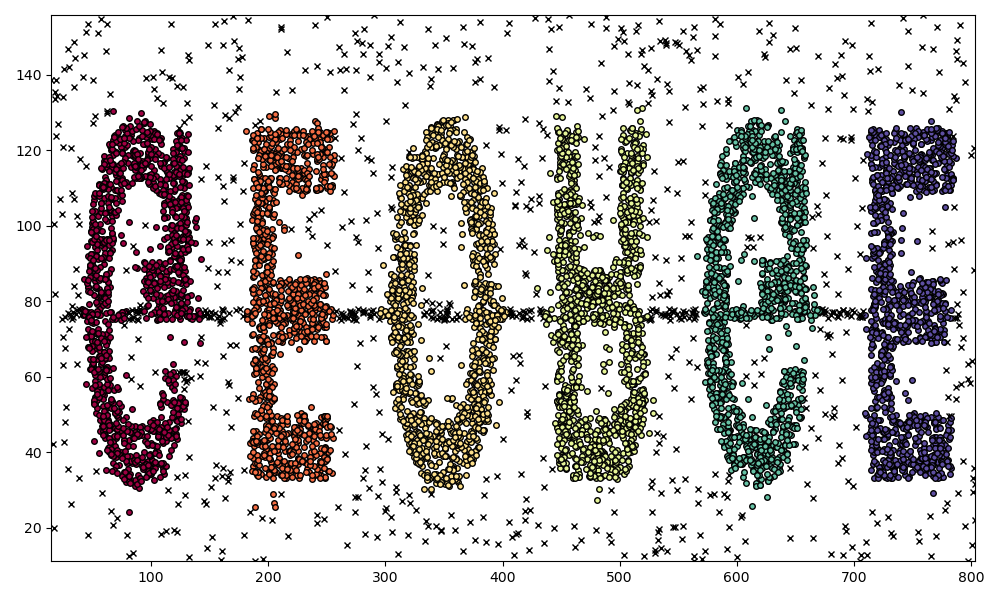}
         \caption{DP-DBSCAN labels (Cluto-t5)}
         \label{fig:cluto_t5_dp_dbscan_labels}
     \end{subfigure} \hfill
     \begin{subfigure}[b]{0.32\textwidth}
         \centering
         \includegraphics[width=\textwidth]{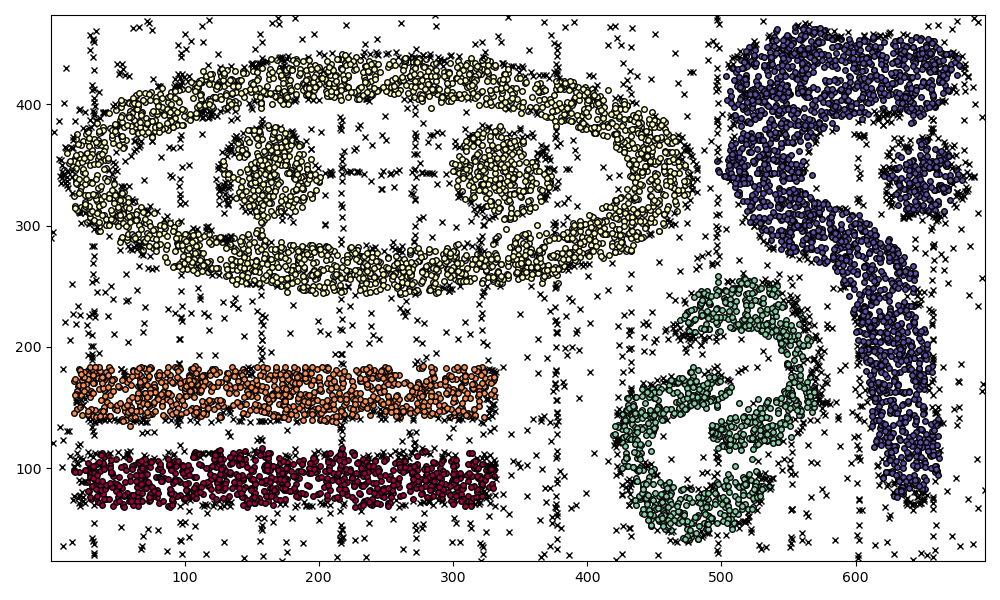}
         \caption{DP-DBSCAN labels (Cluto-t7)}
         \label{fig:cluto_t7_dp_dbscan_labels}
    \end{subfigure}
    
    \begin{subfigure}[b]{0.32\textwidth}
         \centering
         \includegraphics[width=\textwidth]{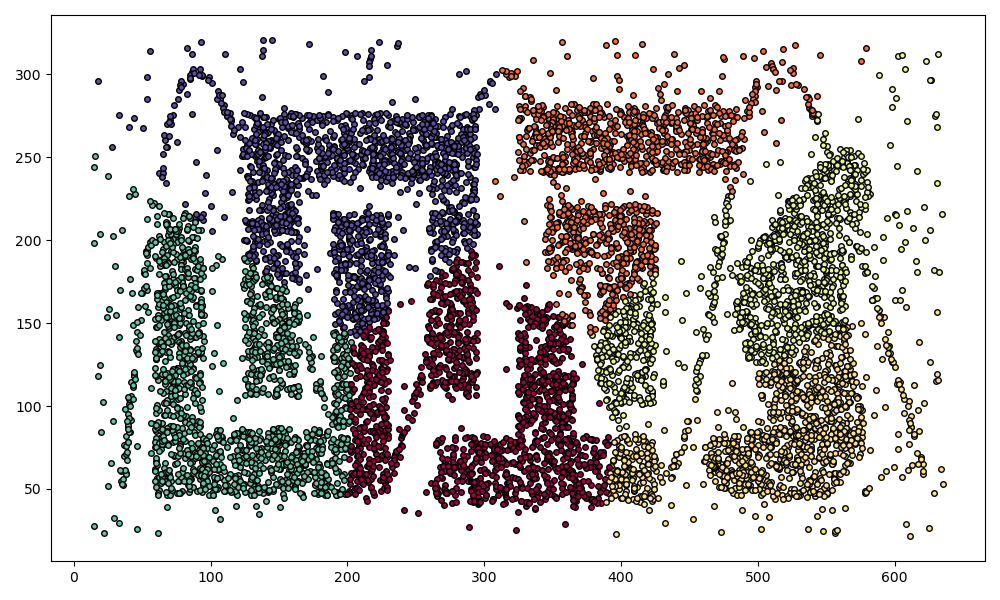}
         \caption{DP-Kmeans labels (Cluto-t4)}
         \label{fig:cluto_t4_dp_kmeans_labels}
     \end{subfigure} \hfill
     \begin{subfigure}[b]{0.32\textwidth}
         \centering
         \includegraphics[width=\textwidth]{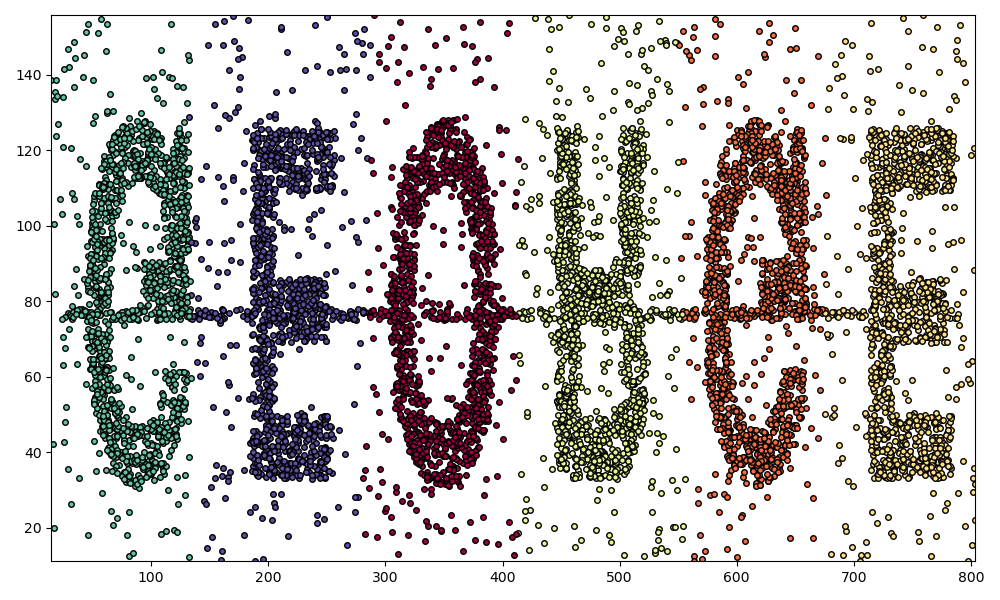}
         \caption{DP-Kmeans labels (Cluto-t5)}
         \label{fig:cluto_t5_dp_kmeans_labels}
     \end{subfigure} \hfill
     \begin{subfigure}[b]{0.32\textwidth}
         \centering
         \includegraphics[width=\textwidth]{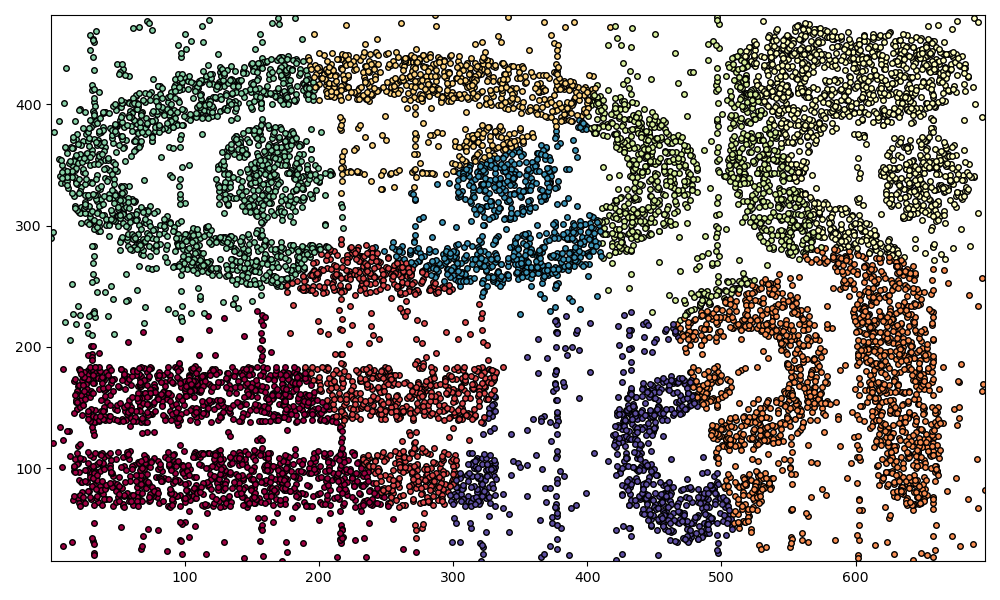}
         \caption{DP-Kmeans labels (Cluto-t7)}
         \label{fig:cluto_t7_dp_kmeans_labels}
    \end{subfigure} 
    
    \caption{Comparison of accuracy on Cluto datasets.}
    \label{fig:cluto}
\end{figure*}

We provide a detailed comparison of accuracy on the Cluto datasets plotted in Figure~\ref{fig:cluto}.
Figure~\ref{fig:cluto_t4_dbscan} to \ref{fig:cluto_t7_dbscan} show the (non-private) clusters produced by DBSCAN, run with parameters specified in Table~\ref{tab:dataset}. The original DBSCAN successfully finds all clusters in the datasets. Figure~\ref{fig:cluto_t4_dp_dbscan} to \ref{fig:cluto_t7_dp_dbscan} show the approximate cluster spans from our differentially private DP-DBSCAN mechanism, run with $\varepsilon=1$ and the same $(\alpha, \MinPts')$ that have been discussed in Section~\ref{sec:parameters}. On Cluto-t4 and Cluto-t5, all the 6 clusters are successfully found by our mechanism. On Cluto-t7 however, DP-DBSCAN reports 5 instead of 9 clusters since some clusters are not well-separated and are merged into one.
This is an expected performance under DP, as two close clusters can easily be merged on a neighboring dataset, with the addition of a single point in between.

To compare, We run DP-Kmeans with the best $k$ that equals to the exact number of clusters on each dataset, under the same $\varepsilon=1$. The $k$ centers and their clusters defined by the Voronoi diagram are shown in Figure~\ref{fig:cluto_t4_dp_kmeans} to \ref{fig:cluto_t7_dp_kmeans}.
It should be clear that approximate cluster spans produced by DP-DBSCAN are better descriptions that capture the shapes of clusters on these complex datasets, especially for Cluto-t4 and t7. In addition, DBSCAN can be used for identifying noises, which cannot be done by k-means. Further, due to the limit of Kmeans in finding clusters of arbitrary shapes, the quality of clusters can hardly be improved by other centroid-based clustering methods under DP  (e.g.~\cite{DBLP:conf/icml/ChangG0M21}), since the same problems exists for the original Kmeans where $\varepsilon=\infty$.

\begin{table}[htbp]
    \centering
    \caption{Comparison of ARI and AMI scores.}
    \label{tab:score}
    \begin{tabular}{ccccccc}
    \toprule
         & & \multicolumn{2}{c}{$\varepsilon=1$} & \multicolumn{2}{c}{$\varepsilon=\infty$} \\

        \multicolumn{2}{c}{Dataset} & DP-DBSCAN & DP-Kmeans & DBSCAN & Kmeans \\
    \midrule
        \multirow{6}{*}{ARI} & Circles & \textbf{0.94} &  0.00 & \textbf{0.98} & 0.00 \\ 
        & Moons & \textbf{0.99} & 0.51 & \textbf{1.00} & 0.47 \\ 
        & Blobs & \textbf{0.81} &  0.79 & 0.55 & \textbf{0.89} \\ 
        & Cluto-t4 & \textbf{0.64} & 0.47 & \textbf{0.95} & 0.50 \\
        & Cluto-t5 & \textbf{0.93} &  0.69 & \textbf{0.96}  & 0.78 \\
        & Cluto-t7 & \textbf{0.52} &  0.32 & \textbf{0.76} & 0.33 \\
        & HAR70+ & \textbf{0.40} & 0.19 & \textbf{0.57} & 0.23 \\
    \midrule
        \multirow{6}{*}{AMI} & Circles & \textbf{0.92} & 0.00 & \textbf{0.96} & 0.00 \\ 
        & Moons & \textbf{0.99} & 0.41 & \textbf{1.00} & 0.37 \\ 
        & Blobs & \textbf{0.83} & 0.79 & 0.66 & \textbf{0.87} \\ 
        & Cluto-t4 & \textbf{0.74} &0.59 & \textbf{0.92} & 0.61 \\
        & Cluto-t5 & \textbf{0.92} &  0.77 & \textbf{0.95}  & 0.82 \\
        & Cluto-t7 & \textbf{0.63} &  0.54 & \textbf{0.82} & 0.56 \\
        & HAR70+ & \textbf{0.45} & 0.43 & \textbf{0.54} & 0.48 \\
    \bottomrule
    \end{tabular}
\end{table}

To give a score of clustering quality, we extract the point labels for both DP-DBSCAN and DP-Kmeans in Figure~\ref{fig:cluto_t4_dp_dbscan_labels} to \ref{fig:cluto_t7_dp_kmeans_labels}.
Recall that for both mechanisms, the point labels are not outputs and do not have a privacy guarantee. It is  just for evaluating the clustering quality in our experiment.
Table~\ref{tab:score} summaries the ARI and AMI scores for the Cluto dataset among other datasets, averaged over 3 trials for the DP mechanisms. We observe that DP-DBSCAN outperforms DP-Kmeans on all the datasets, mainly because the original DBSCAN outperforms Kmeans on most evaluated datasets. Kmeans performs well on Blobs, which is arguably most suitable for centroid-based methods.
Our mechanism is even more accurate than non-DP DBSCAN on the Blobs dataset under the given parameters, as it successfully identified all three blobs while the original DBSCAN merged two of the Blobs into one.

\subsection{Varying Parameters}

In this section, we evaluate the effect of different parameters on accuracy of our mechanism. We use the default $(\alpha, \MinPts)$ with $\varepsilon=1$ and $\eta=4$ unless otherwise specified.

\begin{figure}[htbp]
    \centering
     \begin{subfigure}[b]{0.48\textwidth}
         \centering
         \includegraphics[width=\textwidth]{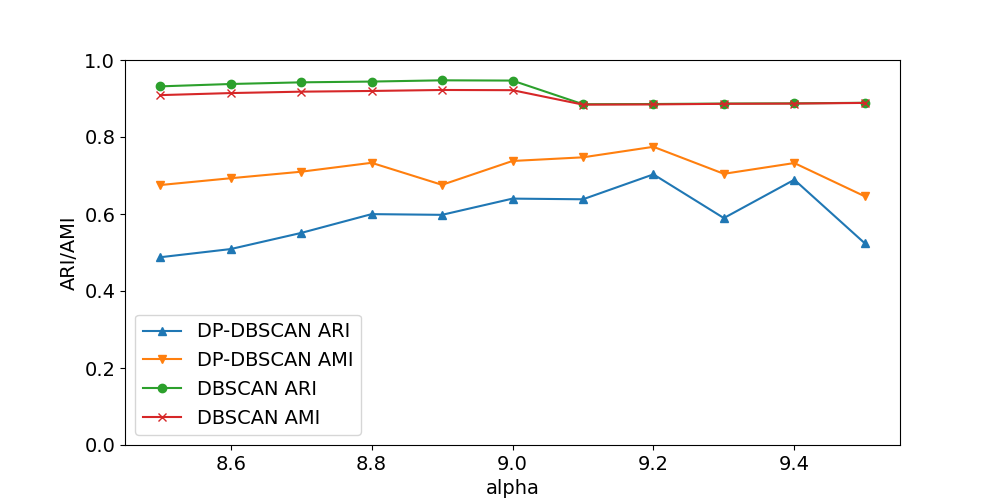}
     \end{subfigure}
     \begin{subfigure}[b]{0.48\textwidth}
         \centering
         \includegraphics[width=\textwidth]{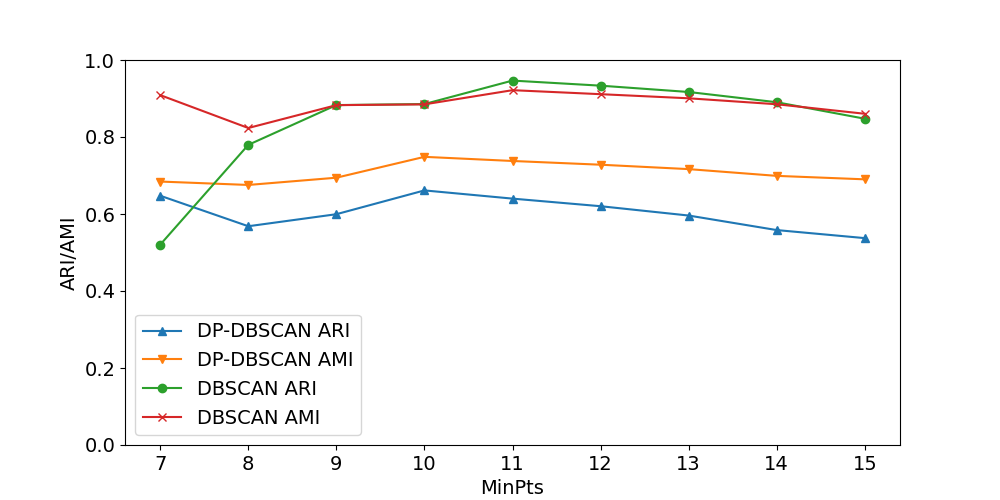}
     \end{subfigure}
    \caption{Accuracy on Cluto-t4, varying $(\alpha, \MinPts)$.}
    \label{fig:dbscan_params}
\end{figure}

\subsubsection{Dependency on DBSCAN parameters ($\alpha, \MinPts$)}

Figure~\ref{fig:dbscan_params} compares the accuracy of DP-DBSCAN against the original DBSCAN on Cluto-t4, varying $\alpha$ and $\MinPts$. Similar to the original DBSCAN, the performance of DP-DBSCAN is robust over a wide range of parameters.
While it is possible that DBSCAN and DP-DBSCAN can be optimized at different $\alpha$ and $\MinPts$, we find that the optimal parameters for DBSCAN will also be good parameter choices for DP-DBSCAN.

Also note that while changing $\alpha$ affects the grid structure, using different $\MinPts$ only affects the post-processing steps of our mechanism without any further privacy costs. Therefore, our mechanism is capable of producing clusters for all values of $\MinPts$ with the $\alpha$ fixed. This can help in an exploratory data analysis, for example, allowing the analyst to decide adaptively the next $\alpha$ to try given clusters at different density levels.

\begin{figure}[htbp]
    \centering
     \begin{subfigure}[b]{0.48\textwidth}
         \centering
         \includegraphics[width=\textwidth]{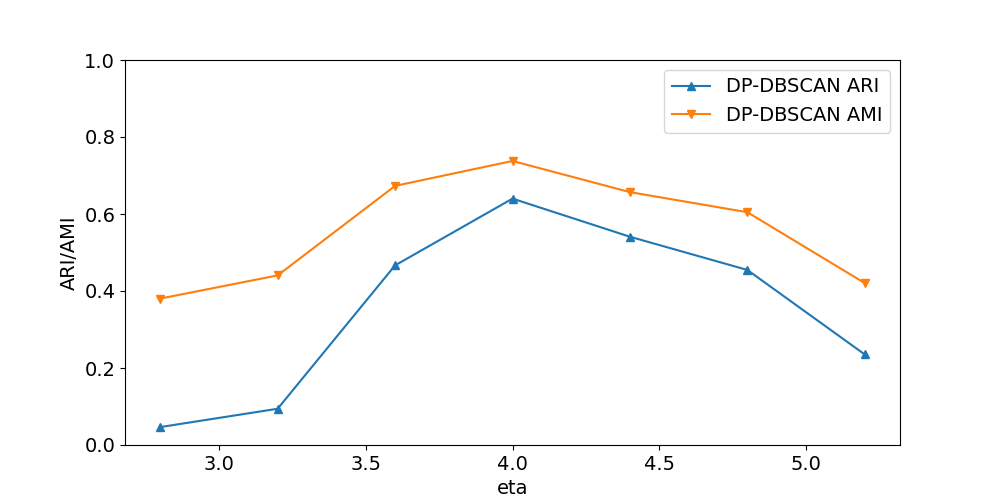}
         \subcaption{Cluto-t4}
     \end{subfigure}
     \begin{subfigure}[b]{0.48\textwidth}
         \centering
         \includegraphics[width=\textwidth]{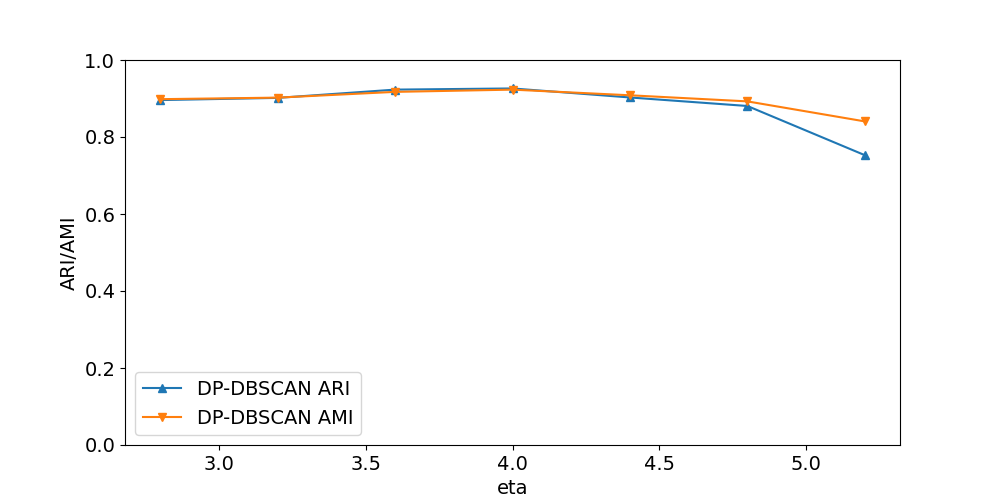}
         \subcaption{Cluto-t5}
     \end{subfigure}
    \caption{Accuracy on Cluto-t4 and t5, varying $\eta$. Larger $\eta$ means larger cells, which increases $\rho$ but decreases $\tau$.}
    \label{fig:grid_width}
\end{figure}

\subsubsection{Dependency on approximation ratio $(\rho,\tau)$}

Recall that both $\rho$ and $\tau$ are related to the parameter $\eta$, which can be controlled by the cell width $w=\frac{\eta\alpha}{4\sqrt{d}}$. With larger $\eta$ we have larger cells, which increases the approximation ratio $\rho$ to the radius, but decreases the additive error $\tau$ since each cell has less neighbors. In Figure~\ref{fig:grid_width}, we use different cell width to test the performance trade-off between $\rho$ and $\tau$. While the best choice may depend on the dataset, it can be observed that when $\eta=4$,  the grid width of $w=\alpha/\sqrt{d}$ tend to be optimal. This is because: 1) for smaller $\eta$, the error $\Gamma$ grows with $\kappa = (1+\frac{8\sqrt{d}}{\eta})^d$, which may dominate the actual counts in a cluster; 2) for $\eta>4$, core points within the same grid may be from different clusters, which will be indistinguishable and bring error.

\begin{figure}[htbp]
    \centering
     \begin{subfigure}[b]{0.48\textwidth}
         \centering
         \includegraphics[width=\textwidth]{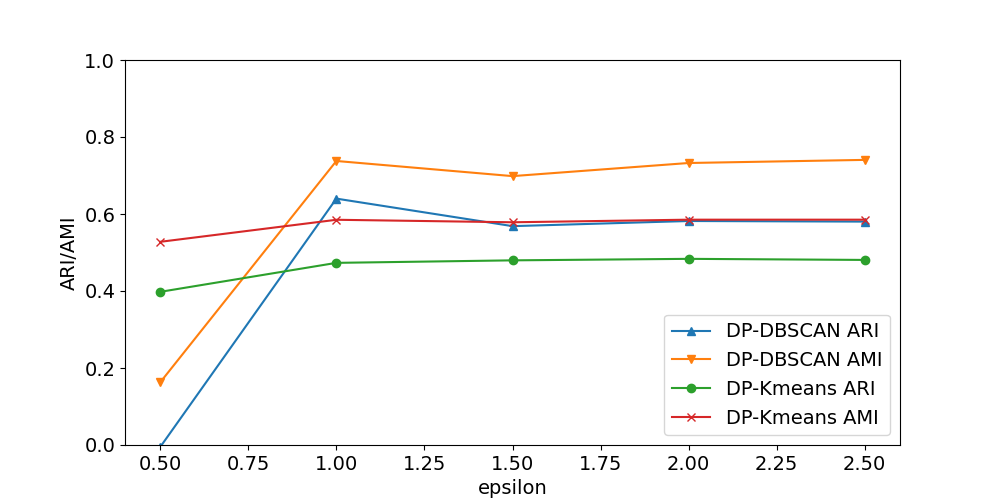}
         \subcaption{Cluto-t4}
     \end{subfigure}
     \begin{subfigure}[b]{0.48\textwidth}
         \centering
        \includegraphics[width=\textwidth]{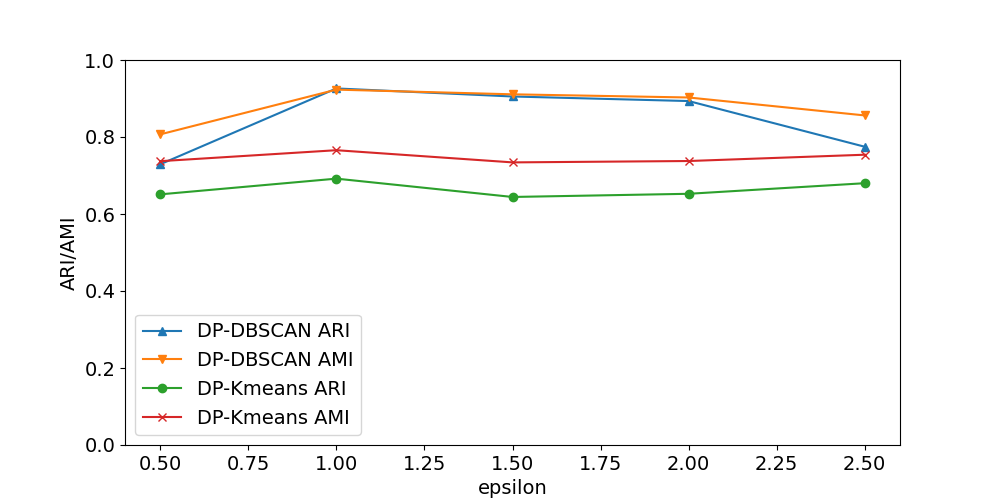}
        \subcaption{Cluto-t5}
    \end{subfigure}
    \caption{Accuracy on Cluto-t4 and t5, varying $\varepsilon$.}
    \label{fig:eps}
\end{figure}

\subsubsection{Dependency on privacy budget $\varepsilon$}

Figure~\ref{fig:eps} compares the accuracy of DP-DBSCAN and DP-KMeans with respect to $\varepsilon$. It is clear that DP-DPSCAN has better accuracy than DP-Kmeans as long as $\varepsilon \geq 1$, but we also observe that the accuracy of DP-DBSCAN drops significantly on Cluto-t4 when $\varepsilon$ goes towards 0, while DP-Kmeans seems to maintain the same level of accuracy. This is because for our mechanism, decreasing $\varepsilon$ from $1.0$ to $0.5$ effectively increases $\tau$ by a factor of 2, which dwarfs the real value of $\MinPts$, corresponding to the true counts of points in the neighborhood. As a result, all points are classified as noises when $\varepsilon=0.5$. On the other hand, since we provide the correct number of clusters to DP-Kmeans, it can always assign the right number of distinct labels, achieving some utility. In fact, even if we set $\varepsilon=1/n^2$ which breaks the utility of even a counting query, DP-Kmeans still show an accuracy of $\ARI=0.34$ and $\AMI=0.50$ when the clusters are almost completely data-independent.

The valid range of $\varepsilon$ depends on the dataset, since it contributes to an additive noise. Naturally, for datasets where a large $\MinPts$ is still valid for producing meaningful clusters, having a small $\varepsilon$ does not cause our mechanism to lose utility.

\subsection{Running Time}\label{sec:time}

\begin{table}[htbp]
    \centering
    \caption{Comparison of Running Time.}
    \label{tab:time}
    \begin{tabular}{cccccc}
    \toprule
        \multirow{2}{*}{Dataset} & \multicolumn{2}{c}{DP-DBSCAN} & \multirow{2}{*}{DP-Kmeans} & \multirow{2}{*}{DBSCAN} & \multirow{2}{*}{Kmeans} \\
        & $O(|\X|)$ & $O(n)$ \\
    \midrule
        Circles & 103 ms & 12.8 ms & 10.1 ms & 5.13 ms & 19.7 ms\\
        Moons & 69.9 ms & 14.9 ms & 9.78 ms & 5.62 ms & 19.4 ms \\
        Blobs & 77.1 ms & 12.4 ms & 7.72 ms & 8.27 ms & 13.6 ms\\
        Cluto-t4 & 349 ms & 18.7 ms & 33.6 ms & 18.6 ms & 91.8 ms  \\
        Cluto-t5 & 213 ms & 57.3 ms & 33.7 ms & 22.3 ms & 87.7 ms\\
        Cluto-t7 & 335 ms & 25.6 ms & 34.3 ms & 26.2 ms & 88.4 ms\\
        Crash & 26.6 s & 5.49 s & 2.57 s & 13.6 s & 349 ms \\
        Cabs &  - & 27.9 s & 15.6 s & -$^\dagger$ & 1.43 s\\
        Cabs-tiny &  5.49 s & 2.19 s & 1.14 s & 9.86 s & 183 ms\\
        HAR70+ &  - & 23.6 s & 334 ms & 729 ms & 120 ms\\
    \bottomrule
    \multicolumn{6}{l}{$^{\dagger}$\footnotesize{out of memory}}
    \end{tabular}
\end{table}

Table~\ref{tab:time} summaries the running time of DP-DBSCAN with both the naive Laplace histogram and the linear time histogram, compared with DP-Kmeans (implemented by IBM) and DBSCAN/Kmeans (implemented in scikit-learn). For 2D datasets, the linear time version has similar time cost as existing mechanisms, and the naive implementation is in general more time-costly.
For the 3D HAR70+ dataset, there is a clear increase in running time, since the constant $\kappa$ depends on $d$ exponentially.
But with the help of the linear-time histogram, our algorithm can still be run within half a minute on this dataset, which has $|\X|=10^8$ cells.
We also find a majority of running time spent on collecting and privatizing the noisy histogram and noisy sums, with the follow-up clustering almost free-of-charge.
Constructing such a histogram may benefit other applications as well.

\section{Acknowledgments}
This work has been supported by HKRGC under grants 16205422, 16204223, and 16203924; and by the National Research Foundation, Prime Minister’s Office, Singapore under its Campus for Research Excellence and Technological Enterprise (CREATE) programme. We would also like to thank the anonymous reviewers who have made valuable suggestions on improving the presentation of the paper.

\bibliographystyle{plainnat}
\bibliography{main.bib}

\begin{thebibliography}{54}
\providecommand{\natexlab}[1]{#1}
\providecommand{\url}[1]{\texttt{#1}}
\expandafter\ifx\csname urlstyle\endcsname\relax
  \providecommand{\doi}[1]{doi: #1}\else
  \providecommand{\doi}{doi: \begingroup \urlstyle{rm}\Url}\fi

\bibitem[Abadi et~al.(2016)Abadi, Chu, Goodfellow, McMahan, Mironov, Talwar, and Zhang]{DBLP:conf/ccs/AbadiCGMMT016}
Mart{\'{\i}}n Abadi, Andy Chu, Ian~J. Goodfellow, H.~Brendan McMahan, Ilya Mironov, Kunal Talwar, and Li~Zhang.
\newblock Deep learning with differential privacy.
\newblock In \emph{{CCS}}, pages 308--318. {ACM}, 2016.

\bibitem[Alabi et~al.(2022)Alabi, McMillan, Sarathy, Smith, and Vadhan]{DBLP:journals/popets/AlabiMSSV22}
Daniel Alabi, Audra McMillan, Jayshree Sarathy, Adam~D. Smith, and Salil~P. Vadhan.
\newblock Differentially private simple linear regression.
\newblock \emph{Proc. Priv. Enhancing Technol.}, 2022\penalty0 (2):\penalty0 184--204, 2022.

\bibitem[Balcan et~al.(2017)Balcan, Dick, Liang, Mou, and Zhang]{DBLP:conf/icml/BalcanDLMZ17}
Maria{-}Florina Balcan, Travis Dick, Yingyu Liang, Wenlong Mou, and Hongyang Zhang.
\newblock Differentially private clustering in high-dimensional euclidean spaces.
\newblock In \emph{{ICML}}, volume~70 of \emph{Proceedings of Machine Learning Research}, pages 322--331. {PMLR}, 2017.

\bibitem[Balcer and Vadhan(2019)]{DBLP:journals/jpc/BalcerV19}
Victor Balcer and Salil~P. Vadhan.
\newblock Differential privacy on finite computers.
\newblock \emph{J. Priv. Confidentiality}, 9\penalty0 (2), 2019.

\bibitem[Bozdemir et~al.(2021)Bozdemir, Canard, Ermis, M{\"{o}}llering, {\"{O}}nen, and Schneider]{DBLP:conf/asiaccs/BozdemirCEMO021}
Beyza Bozdemir, S{\'{e}}bastien Canard, Orhan Ermis, Helen M{\"{o}}llering, Melek {\"{O}}nen, and Thomas Schneider.
\newblock Privacy-preserving density-based clustering.
\newblock In \emph{AsiaCCS}, pages 658--671. {ACM}, 2021.

\bibitem[Campello et~al.(2013)Campello, Moulavi, and Sander]{DBLP:conf/pakdd/CampelloMS13}
Ricardo J. G.~B. Campello, Davoud Moulavi, and J{\"{o}}rg Sander.
\newblock Density-based clustering based on hierarchical density estimates.
\newblock In \emph{{PAKDD} {(2)}}, volume 7819 of \emph{Lecture Notes in Computer Science}, pages 160--172. Springer, 2013.

\bibitem[Campello et~al.(2015)Campello, Moulavi, Zimek, and Sander]{DBLP:journals/tkdd/CampelloMZS15}
Ricardo J. G.~B. Campello, Davoud Moulavi, Arthur Zimek, and J{\"{o}}rg Sander.
\newblock Hierarchical density estimates for data clustering, visualization, and outlier detection.
\newblock \emph{{ACM} Trans. Knowl. Discov. Data}, 10\penalty0 (1):\penalty0 5:1--5:51, 2015.

\bibitem[Chan et~al.(2010)Chan, Shi, and Song]{DBLP:conf/icalp/ChanSS10}
T.{-}H.~Hubert Chan, Elaine Shi, and Dawn Song.
\newblock Private and continual release of statistics.
\newblock In \emph{{ICALP} {(2)}}, volume 6199 of \emph{Lecture Notes in Computer Science}, pages 405--417. Springer, 2010.

\bibitem[Chan et~al.(2012)Chan, Shi, and Song]{DBLP:conf/esa/ChanSS12}
T.{-}H.~Hubert Chan, Elaine Shi, and Dawn Song.
\newblock Optimal lower bound for differentially private multi-party aggregation.
\newblock In \emph{{ESA}}, volume 7501 of \emph{Lecture Notes in Computer Science}, pages 277--288. Springer, 2012.

\bibitem[Chang et~al.(2021)Chang, Ghazi, Kumar, and Manurangsi]{DBLP:conf/icml/ChangG0M21}
Alisa Chang, Badih Ghazi, Ravi Kumar, and Pasin Manurangsi.
\newblock Locally private k-means in one round.
\newblock In \emph{{ICML}}, volume 139 of \emph{Proceedings of Machine Learning Research}, pages 1441--1451. {PMLR}, 2021.

\bibitem[Cicirello(2024)]{DBLP:journals/corr/abs-2403-11018}
Vincent~A. Cicirello.
\newblock On the average runtime of an open source binomial random variate generation algorithm.
\newblock \emph{CoRR}, abs/2403.11018, 2024.

\bibitem[Cohen et~al.(2021)Cohen, Kaplan, Mansour, Stemmer, and Tsfadia]{DBLP:conf/icml/CohenKMST21}
Edith Cohen, Haim Kaplan, Yishay Mansour, Uri Stemmer, and Eliad Tsfadia.
\newblock Differentially-private clustering of easy instances.
\newblock In \emph{{ICML}}, volume 139 of \emph{Proceedings of Machine Learning Research}, pages 2049--2059. {PMLR}, 2021.

\bibitem[Cormode et~al.(2012)Cormode, Procopiuc, Srivastava, and Tran]{DBLP:conf/icdt/CormodePST12}
Graham Cormode, Cecilia~M. Procopiuc, Divesh Srivastava, and Thanh T.~L. Tran.
\newblock Differentially private summaries for sparse data.
\newblock In \emph{{ICDT}}, pages 299--311. {ACM}, 2012.

\bibitem[Dua and Graff(2017)]{Dua:2019}
Dheeru Dua and Casey Graff.
\newblock Uci machine learning repository, 2017.
\newblock URL \url{http://archive.ics.uci.edu}.

\bibitem[Dwork and Roth(2014)]{DBLP:journals/fttcs/DworkR14}
Cynthia Dwork and Aaron Roth.
\newblock The algorithmic foundations of differential privacy.
\newblock \emph{Found. Trends Theor. Comput. Sci.}, 9\penalty0 (3-4):\penalty0 211--407, 2014.

\bibitem[Dwork et~al.(2010)Dwork, Naor, Pitassi, and Rothblum]{DBLP:conf/stoc/DworkNPR10}
Cynthia Dwork, Moni Naor, Toniann Pitassi, and Guy~N. Rothblum.
\newblock Differential privacy under continual observation.
\newblock In \emph{{STOC}}, pages 715--724. {ACM}, 2010.

\bibitem[Ester et~al.(1996)Ester, Kriegel, Sander, and Xu]{DBLP:conf/kdd/EsterKSX96}
Martin Ester, Hans{-}Peter Kriegel, J{\"{o}}rg Sander, and Xiaowei Xu.
\newblock A density-based algorithm for discovering clusters in large spatial databases with noise.
\newblock In \emph{{KDD}}, pages 226--231. {AAAI} Press, 1996.

\bibitem[Fu et~al.(2024{\natexlab{a}})Fu, Cheng, Chang, and Shen]{Fu2024PPADBSCANP}
Jiaxuan Fu, Ke~Cheng, Zhao Chang, and Yulong Shen.
\newblock Ppa-dbscan: Privacy-preserving $\rho$-approximate density-based clustering.
\newblock \emph{IEEE Transactions on Dependable and Secure Computing}, 2024{\natexlab{a}}.

\bibitem[Fu et~al.(2024{\natexlab{b}})Fu, Cheng, Song, Xia, Chang, and Shen]{DBLP:journals/tifs/FuCSXCS24}
Jiaxuan Fu, Ke~Cheng, Anxiao Song, Yuheng Xia, Zhao Chang, and Yulong Shen.
\newblock {FSS-DBSCAN:} outsourced private density-based clustering via function secret sharing.
\newblock \emph{{IEEE} Trans. Inf. Forensics Secur.}, 19:\penalty0 7759--7773, 2024{\natexlab{b}}.

\bibitem[Gan and Tao(2015)]{DBLP:conf/sigmod/GanT15}
Junhao Gan and Yufei Tao.
\newblock {DBSCAN} revisited: Mis-claim, un-fixability, and approximation.
\newblock In \emph{{SIGMOD} Conference}, pages 519--530. {ACM}, 2015.

\bibitem[Gentle(2003)]{gentle2003random}
James~E Gentle.
\newblock \emph{Random number generation and Monte Carlo methods}, volume 381.
\newblock Springer, 2003.

\bibitem[Ghazi et~al.(2020)Ghazi, Kumar, and Manurangsi]{DBLP:conf/nips/Ghazi0M20}
Badih Ghazi, Ravi Kumar, and Pasin Manurangsi.
\newblock Differentially private clustering: Tight approximation ratios.
\newblock In \emph{NeurIPS}, 2020.

\bibitem[Ghosh et~al.(2012)Ghosh, Roughgarden, and Sundararajan]{DBLP:journals/siamcomp/GhoshRS12}
Arpita Ghosh, Tim Roughgarden, and Mukund Sundararajan.
\newblock Universally utility-maximizing privacy mechanisms.
\newblock \emph{{SIAM} J. Comput.}, 41\penalty0 (6):\penalty0 1673--1693, 2012.

\bibitem[Hahsler et~al.(2019)Hahsler, Piekenbrock, and Doran]{JSSv091i01}
Michael Hahsler, Matthew Piekenbrock, and Derek Doran.
\newblock {DBSCAN}: Fast density-based clustering with r.
\newblock \emph{Journal of Statistical Software}, 91\penalty0 (1):\penalty0 1–30, 2019.

\bibitem[Hegde et~al.(2021)Hegde, M{\"{o}}llering, Schneider, and Yalame]{DBLP:journals/popets/HegdeMSY21}
Aditya Hegde, Helen M{\"{o}}llering, Thomas Schneider, and Hossein Yalame.
\newblock Sok: Efficient privacy-preserving clustering.
\newblock \emph{Proc. Priv. Enhancing Technol.}, 2021\penalty0 (4):\penalty0 225--248, 2021.

\bibitem[Holohan et~al.(2019)Holohan, Braghin, Mac~Aonghusa, and Levacher]{diffprivlib}
Naoise Holohan, Stefano Braghin, P{\'o}l Mac~Aonghusa, and Killian Levacher.
\newblock Diffprivlib: the {IBM} differential privacy library.
\newblock \emph{ArXiv e-prints}, 1907.02444 [cs.CR], July 2019.

\bibitem[Huang and Liu(2018)]{DBLP:conf/pods/HuangL18}
Zhiyi Huang and Jinyan Liu.
\newblock Optimal differentially private algorithms for k-means clustering.
\newblock In \emph{{PODS}}, pages 395--408. {ACM}, 2018.

\bibitem[Hubert and Arabie(1985)]{hubert1985comparing}
Lawrence Hubert and Phipps Arabie.
\newblock Comparing partitions.
\newblock \emph{Journal of classification}, 2:\penalty0 193--218, 1985.

\bibitem[Jin and Li(2019)]{jin2019improved}
Yuzhen Jin and Shuyu Li.
\newblock An improved differentially private dbscan clustering algorithm for vehicular crowdsensing.
\newblock In \emph{2019 IEEE 13th International Conference on Anti-counterfeiting, Security, and Identification (ASID)}, pages 51--55. IEEE, 2019.

\bibitem[Jones et~al.(2021)Jones, Nguyen, and Nguyen]{DBLP:conf/aaai/JonesNN21}
Matthew Jones, Huy~L. Nguyen, and Thy~Dinh Nguyen.
\newblock Differentially private clustering via maximum coverage.
\newblock In \emph{{AAAI}}, pages 11555--11563. {AAAI} Press, 2021.

\bibitem[Kachitvichyanukul and Schmeiser(1988)]{DBLP:journals/cacm/KachitvichyanukulS88}
Voratas Kachitvichyanukul and Bruce~W. Schmeiser.
\newblock Binomial random variate generation.
\newblock \emph{Commun. {ACM}}, 31\penalty0 (2):\penalty0 216--222, 1988.

\bibitem[Karypis(2002)]{karypis2002cluto}
George Karypis.
\newblock Cluto a clustering toolkit.
\newblock Technical report, Dept. of Computer Science, University of Minnesota, 2002.

\bibitem[Karypis et~al.(1999)Karypis, Han, and Kumar]{DBLP:journals/computer/KarypisHK99}
George Karypis, Eui{-}Hong Han, and Vipin Kumar.
\newblock Chameleon: Hierarchical clustering using dynamic modeling.
\newblock \emph{Computer}, 32\penalty0 (8):\penalty0 68--75, 1999.

\bibitem[Kifer and Machanavajjhala(2011)]{DBLP:conf/sigmod/KiferM11}
Daniel Kifer and Ashwin Machanavajjhala.
\newblock No free lunch in data privacy.
\newblock In \emph{{SIGMOD} Conference}, pages 193--204. {ACM}, 2011.

\bibitem[Lebeda and Tetek(2023)]{DBLP:conf/pods/LebedaT23}
Christian~Janos Lebeda and Jakub Tetek.
\newblock Better differentially private approximate histograms and heavy hitters using the misra-gries sketch.
\newblock In \emph{{PODS}}, pages 79--88. {ACM}, 2023.

\bibitem[Liu et~al.(2013)Liu, Xiong, Luo, and Huang]{DBLP:journals/tdp/LiuXLH13}
Jinfei Liu, Li~Xiong, Jun Luo, and Joshua~Zhexue Huang.
\newblock Privacy preserving distributed {DBSCAN} clustering.
\newblock \emph{Trans. Data Priv.}, 6\penalty0 (1):\penalty0 69--85, 2013.

\bibitem[Lloyd(1982)]{DBLP:journals/tit/Lloyd82}
Stuart~P. Lloyd.
\newblock Least squares quantization in {PCM}.
\newblock \emph{{IEEE} Trans. Inf. Theory}, 28\penalty0 (2):\penalty0 129--136, 1982.

\bibitem[ming Wu and kun Huang(2015)]{wu2015dp}
Wei ming Wu and Huan kun Huang.
\newblock A dp-dbscan clustering algorithm based on differential privacy preserving.
\newblock \emph{Computer Engineering and Science}, 37\penalty0 (04):\penalty0 830, 2015.

\bibitem[Mo et~al.(2024)Mo, Song, and Ding]{DBLP:journals/pacmmod/MoSD24}
Guanlin Mo, Shihong Song, and Hu~Ding.
\newblock Towards metric {DBSCAN:} exact, approximate, and streaming algorithms.
\newblock \emph{Proc. {ACM} Manag. Data}, 2\penalty0 (3):\penalty0 178, 2024.

\bibitem[Nguyen et~al.(2021)Nguyen, Chaturvedi, and Xu]{DBLP:conf/aaai/NguyenCX21}
Huy~L. Nguyen, Anamay Chaturvedi, and Eric~Z. Xu.
\newblock Differentially private k-means via exponential mechanism and max cover.
\newblock In \emph{{AAAI}}, pages 9101--9108. {AAAI} Press, 2021.

\bibitem[Nguyen et~al.(2009)Nguyen, Epps, and Bailey]{DBLP:conf/icml/NguyenEB09}
Xuan~Vinh Nguyen, Julien Epps, and James Bailey.
\newblock Information theoretic measures for clusterings comparison: is a correction for chance necessary?
\newblock In \emph{{ICML}}, volume 382 of \emph{{ACM} International Conference Proceeding Series}, pages 1073--1080. {ACM}, 2009.

\bibitem[Ni et~al.(2017)Ni, Li, Liu, Bourgeois, and Yu]{DBLP:conf/iiki/NiLLBY17}
Lina Ni, Chao Li, Haoran Liu, Anu~G. Bourgeois, and Jiguo Yu.
\newblock Differential private preservation multi-core dbscan clustering for network user data.
\newblock In \emph{{IIKI}}, volume 129 of \emph{Procedia Computer Science}, pages 257--262. Elsevier, 2017.

\bibitem[Ni et~al.(2018)Ni, Li, Wang, Jiang, and Yu]{DBLP:journals/access/NiLWJY18}
Lina Ni, Chao Li, Xiao Wang, Honglu Jiang, and Jiguo Yu.
\newblock {DP-MCDBSCAN:} differential privacy preserving multi-core {DBSCAN} clustering for network user data.
\newblock \emph{{IEEE} Access}, 6:\penalty0 21053--21063, 2018.

\bibitem[OpenData(2024)]{DataCrash}
NYC OpenData.
\newblock {Motor Vehicle Collisions - Crashes}.
\newblock \url{https://data.cityofnewyork.us/Public-Safety/Motor-Vehicle-Collisions-Crashes/h9gi-nx95/about_data}, 2024.

\bibitem[Pedregosa et~al.(2011)Pedregosa, Varoquaux, Gramfort, Michel, Thirion, Grisel, Blondel, Prettenhofer, Weiss, Dubourg, Vanderplas, Passos, Cournapeau, Brucher, Perrot, and Duchesnay]{scikit-learn}
F.~Pedregosa, G.~Varoquaux, A.~Gramfort, V.~Michel, B.~Thirion, O.~Grisel, M.~Blondel, P.~Prettenhofer, R.~Weiss, V.~Dubourg, J.~Vanderplas, A.~Passos, D.~Cournapeau, M.~Brucher, M.~Perrot, and E.~Duchesnay.
\newblock Scikit-learn: Machine learning in {P}ython.
\newblock \emph{Journal of Machine Learning Research}, 12:\penalty0 2825--2830, 2011.

\bibitem[Piorkowski et~al.(2022)Piorkowski, Sarafijanovic-Djukic, and Grossglauser]{c7j010-22}
Michal Piorkowski, Natasa Sarafijanovic-Djukic, and Matthias Grossglauser.
\newblock Crawdad epfl/mobility.
\newblock \url{https://ieee-dataport.org/open-access/crawdad-epflmobility}, 2022.

\bibitem[Schubert et~al.(2017)Schubert, Sander, Ester, Kriegel, and Xu]{DBLP:journals/tods/SchubertSEKX17}
Erich Schubert, J{\"{o}}rg Sander, Martin Ester, Hans{-}Peter Kriegel, and Xiaowei Xu.
\newblock {DBSCAN} revisited, revisited: Why and how you should (still) use {DBSCAN}.
\newblock \emph{{ACM} Trans. Database Syst.}, 42\penalty0 (3):\penalty0 19:1--19:21, 2017.

\bibitem[Senekane(2019)]{senekane2019differentially}
Makhamisa Senekane.
\newblock Differentially private image classification using support vector machine and differential privacy.
\newblock \emph{Machine Learning and Knowledge Extraction}, 1\penalty0 (1):\penalty0 483--491, 2019.

\bibitem[Stemmer and Kaplan(2018)]{DBLP:conf/nips/StemmerK18}
Uri Stemmer and Haim Kaplan.
\newblock Differentially private k-means with constant multiplicative error.
\newblock In \emph{NeurIPS}, pages 5436--5446, 2018.

\bibitem[Su et~al.(2016)Su, Cao, Li, Bertino, and Jin]{DBLP:conf/codaspy/SuCLBJ16}
Dong Su, Jianneng Cao, Ninghui Li, Elisa Bertino, and Hongxia Jin.
\newblock Differentially private k-means clustering.
\newblock In \emph{{CODASPY}}, pages 26--37. {ACM}, 2016.

\bibitem[Suresh(2019)]{DBLP:conf/nips/Suresh19}
Ananda~Theertha Suresh.
\newblock Differentially private anonymized histograms.
\newblock In \emph{NeurIPS}, pages 7969--7979, 2019.

\bibitem[Ustad et~al.(2023)Ustad, Logacjov, Trolleb{\o}, Thingstad, Vereijken, Bach, and Skj{\ae}ret{-}Maroni]{DBLP:journals/sensors/UstadLTTVBS23}
Astrid Ustad, Aleksej Logacjov, Stine~{\O}verengen Trolleb{\o}, Pernille Thingstad, Beatrix Vereijken, Kerstin Bach, and Nina Skj{\ae}ret{-}Maroni.
\newblock Validation of an activity type recognition model classifying daily physical behavior in older adults: The {HAR70+} model.
\newblock \emph{Sensors}, 23\penalty0 (5):\penalty0 2368, 2023.

\bibitem[Wang et~al.(2023)Wang, Li, Chen, Zhang, and Hao]{DBLP:conf/globecom/WangLCZH23}
Yingzhe Wang, Hongwei Li, Hanxiao Chen, Xilin Zhang, and Meng Hao.
\newblock Practical and privacy-preserving density-based clustering via shuffling.
\newblock In \emph{{GLOBECOM}}, pages 50--55. {IEEE}, 2023.

\bibitem[Xu et~al.(2013)Xu, Zhang, Xiao, Yang, Yu, and Winslett]{DBLP:journals/vldb/XuZXYYW13}
Jia Xu, Zhenjie Zhang, Xiaokui Xiao, Yin Yang, Ge~Yu, and Marianne Winslett.
\newblock Differentially private histogram publication.
\newblock \emph{{VLDB} J.}, 22\penalty0 (6):\penalty0 797--822, 2013.

\end{thebibliography}
\end{document}